\pdfoutput=1

\documentclass[10pt,twocolumn,twoside]{IEEEtran}
\usepackage{enumerate}   
\usepackage{here}
\usepackage{subfigure}
\usepackage{citesort}
\usepackage[fleqn]{amsmath}
\usepackage{ascmac}
\usepackage{cite}
\usepackage[dvips]{graphicx}
\usepackage{psfrag}
\usepackage{amsmath}
\usepackage{amsbsy}
\usepackage{amssymb}
\usepackage{latexsym}
\usepackage{color}
\usepackage{algorithm}
\usepackage{algorithmic}

\usepackage{amsthm}

\newtheorem{proposition}{Proposition}
\newtheorem{example}{Example}

\newtheorem{theorem}{Theorem}
\newtheorem{lemma}{Lemma}

\newtheorem{assumption}{Assumption}
\newtheorem{fact}{Fact}
\newtheorem{remark}{Remark}

\newcommand{\sgn}{{\rm sgn}}

\newcommand{\euclidspace}{{\mathcal{H}}}
\newcommand{\signal}[1]{{\boldsymbol{#1}}}
\newcommand{\Natural}{{\mathbb N}}
\newcommand{\Naturalp}{{\mathbb N}^{\ast}}
\newcommand{\norm}[1]{\left\|#1\right\|}
\newcommand{\abs}[1]{\left|#1\right|}

\newcommand{\real}{{\mathbb R}}
\newcommand{\Real}{{\mathbb R}}

\newcommand{\innerprod}[2]{\left\langle{#1},{#2}\right\rangle}

\newcommand{\refeq}[1]{(\ref{#1})}
\newcommand{\reftab}[1]{Table \ref{#1}}
\newcommand{\reffig}[1]{Figure \ref{#1}}

\newcommand{\argmin}{\operatornamewithlimits{argmin}}

\newcommand{\X}{{\mathcal{X}}}

\renewcommand{\H}{{\mathcal{H}}}
\newcommand{\M}{{\mathcal{M}}}

\newcommand{\mathO}{\mathcal{O}}
\newcommand{\J}{{\mathcal{J}}}

\newcommand{\T}{{\sf T}}
\newcommand{\Dict}{{\mathcal{D}}}

\newcommand{\di}{{\rm d}}
\newcommand{\sigu}{\signal{u}}
\newcommand{\sigh}{\signal{h}}
\newcommand{\sigw}{\signal{w}}

\newcommand{\sigf}{\signal{f}}

\newcommand{\sigx}{\signal{x}}
\newcommand{\sigy}{\signal{y}}
\newcommand{\sigv}{\signal{v}}
\newcommand{\sigz}{\signal{z}}

\newcommand{\sigg}{\signal{g}}

\newcommand{\sigG}{\signal{G}}

\newcommand{\sigR}{\signal{R}}
\newcommand{\sigp}{\signal{p}}

\newcommand{\sigA}{\signal{A}}
\newcommand{\sigI}{\signal{I}}
\newcommand{\sigP}{\signal{P}}
\newcommand{\x}{\times}
\newcommand{\s}{\sigma}

\newcommand{\ka}[1]{\kappa\left(#1\right)}

\newcommand{\inpro}[1]{\left<#1\right>}
\newcommand{\Ex}[1]{E\left[#1\right]}

\definecolor{darkgreen}{rgb}{0,.6,0}
\definecolor{medorange}{rgb}{0.7,0.3,0}
\definecolor{cyancyan}{rgb}{0.68, 0.92, 0.92}
\def\nn{\nonumber}

\begin{document}

\title{Online Nonlinear Estimation via Iterative $L^2$-Space
Projections: Reproducing Kernel of Subspace}
%
%
 \author{Motoya Ohnishi,~\IEEEmembership{Student Member,~IEEE},
and Masahiro Yukawa,~\IEEEmembership{Member,~IEEE}
\thanks{This work was partially presented at EUSIPCO 2017 \cite{ohnishi17}.}
\thanks{This work was supported by 
JSPS Grants-in-Aid (15K06081, 15K13986, 15H02757) and Scandinavia-Japan Sasakawa Foundation.}
\thanks{The real GPS data used in the numerical experiment is recorded by Isaac Skog at KTH Royal Institute of Technology.}
\thanks{M. Ohnishi is with the Department of 
Electronics and Electrical Engineering,
Keio University, Japan (e-mail: ohnishi@ykw.elec.keio.ac.jp).
}
\thanks{M. Yukawa is with the Department of Electronics and Electrical Engineering, Keio University, Japan (email: yukawa@elec.keio.ac.jp), and with the Center for Advanced Intelligence Project, RIKEN, Japan.
}
}

\maketitle

\begin{abstract}
We propose a novel online learning paradigm 
for nonlinear-function estimation tasks 
based on the iterative projections in the $L^2$ space 
with probability measure reflecting the stochastic property of input signals.
The proposed learning algorithm exploits the reproducing kernel of the
so-called dictionary subspace,
based on the fact that any finite-dimensional space of functions has a
 reproducing kernel characterized by the Gram matrix.
The $L^2$-space geometry provides the best decorrelation property in principle.
The proposed learning paradigm is
significantly different
from the conventional kernel-based learning paradigm
in two senses: (i) the whole space is {\em not} a reproducing
kernel Hilbert space and (ii) the minimum mean squared error estimator
gives the best approximation of the desired nonlinear function in the
dictionary subspace.
It preserves efficiency in computing the inner product as well as 
in updating the Gram matrix when the dictionary grows.
Monotone approximation,
asymptotic optimality, and convergence of the proposed algorithm are analyzed
based on the variable-metric version of adaptive projected subgradient method.
Numerical examples show the efficacy of the proposed algorithm for real data
over a variety of methods including 
the extended Kalman filter and many batch machine-learning methods such as the multilayer perceptron.
\end{abstract}


%

\section{Introduction}
\label{sec:introduction}

\subsection{Background}
Metric is a dominant factor in controlling convergence behaviors
of online learning algorithms,
as witnessed by the extensive studies on adaptive filtering
\cite{amari1998natural,ysy_apqp2007,variable1,narayan1983transform,duttweiler2000proportionate,benesty2002improved,yukawa_ieeesp09}
as well as the recent advances in stochastic optimization
\cite{Adagrad,Adam,Adadelta,ushioICASSP2017}
(see also \cite[Chapter 3]{shor2012minimization} for a related idea called 
{\em space dilation} for accelerating the convergence of the subgradient
method for minimization of nondifferentiable functions).
Metric projection has been used extensively in adaptive/online learning
algorithms
\cite{NLMS1, APA1, haykin, sayed2003, yamada2002efficient,
	yukawaIEICE2004, yukawaTSP2006POWER, yukawaIEICE2010MultiDomain}
(see also the tutorial paper
\cite{thslya_IEEESPMAG11}).
The main subject of the present study is the metric of online learning
algorithms for {\em nonlinear}-function estimation tasks.

Kernel adaptive filtering is a powerful approach to
the nonlinear estimation tasks
\cite{csato01,kivinen04,engel04,huang_icassp05,laskov06,liu_TSP08,slavakis08,richard09,liu_book10,vaerenbergh12,chen_TNNLS12,tayu_tsp15,tayu_tsp16,YukawaRIMS,2013mixture,zhao2016self,chen2016generalized,ma2017robust,zhao2015kernel,scardapane2015online},
being an adaptive extension of
the kernel ridge regression \cite{muller01intro, schoelkopf2002} or
Gaussian process \cite{rasmussen2006gaussian}.
Projection-based kernel adaptive filtering algorithms have been studied
mainly by casting the nonlinear estimation as a minimization problem
either (i) in the Euclidean space of coefficient vectors \cite{richard09},
or (ii) in the reproducing kernel Hilbert space (RKHS)
\cite{slavakis08,liu_book10,tayu_tsp15,tayu_tsp16}.
The two types of formulation induce two different geometries.
The latter type is referred to as the functional approach, and
its geometry in the dictionary subspace (i.e., the subspace spanned by
the dictionary) can be expressed in the Euclidean space equivalently
{\em with a metric characterized by the kernel matrix} \cite{tayu_tsp16}.
The functional approach tends to exhibit better convergence behaviors
(see, e.g., \cite{tayu_tsp15,tayu_tsp16,yukawa_tsp15})
than the former approach.
This has been supported theoretically in \cite{yukawaletter16}.
Specifically, 
provided that the dictionary can be considered
as a set of realizations of the input vectors,
the autocorrelation matrix can be approximated
by a squared kernel matrix essentially,
which indicates that its eigenvalue spread 
for the functional approach is reduced to
a square root compared to the former approach
in principle.
The conventional kernel adaptive filtering methods employ a single
kernel, thereby working efficiently only when all the three conditions
are satisfied:
(i) the target nonlinear function is sufficiently {\em simple},
(ii) its scale is known prior to adaptation 
so that one can design a Gaussian kernel with appropriate scale, and
(iii) the scale is time-invariant.

Multikernel adaptive filtering 
\cite{yukawa_tsp12,yukawa_eusipco13,yukawa_tsp15,toda2017} 
is an efficient solution to the case 
in which some of the above conditions are violated, such as the case
of multi-component/partially-linear functions (see \cite{yukawa_tsp15}).
A remarkable feature of multikernel adaptive filtering is that 
finding a well-fitting kernel and 
obtaining a compact representation (i.e., dictionary sparsification and
parameter estimation) are simultaneously achieved
within a convex analytic framework.
The existing functional approach for multikernel adaptive filtering
is called the Cartesian hyperplane projection along affine subspace
(CHYPASS) algorithm \cite{yukawa_tsp15}, formulated in the Cartesian product
of the RKHSs associated with the multiple kernels employed.
Here, CHYPASS is a multikernel extension of 
the hyperplane projection along affine subspace (HYPASS) algorithm \cite{hypass,tayu_tsp15},
which is an efficient functional approach derived by formulating
the normalized least mean square (NLMS) algorithm in the functional subspace.
The decorrelation property of CHYPASS is however suboptimal
since it counts no correlations among different kernels.

\subsection{Motivation and Contributions}
Suppose that the input (sample) is a real random vector.
Our first primitive question is the following:
{\em what metric induces the best geometry
	having a perfect decorrelation property
	for online nonlinear-function estimation over a (possibly expanding) finite-dimensional subspace in general?}
An immediate answer to this question is the $L^2$ space (the set of square-integrable functions)
under the probability measure determined by
the probability density function of the input vector
(see Sections \ref{subsec:sysmodel} and \ref{subsec:whyhow}).
Henceforth, we simply call it the $L^2$ space.
In addition to its nice geometric property, the $L^2$ space is sufficiently large
to accomodate the subspace even if it expands as time goes by 
(see Section \ref{subsec:sysmodel}).
The $L^2$ space, however, is not an RKHS because
the function value at some specific point is not well-defined
due to the presence of equivalence class.
Now arises the central question penetrating this paper.

{\em Should the learning space be an RKHS to achieve
	efficient online nonlinear estimation?}

In this paper, 
we propose an efficient online nonlinear-function learning paradigm
based on iterative projections in the $L^2$ space.
In the proposed learning paradigm,
the minimum mean squared error (MMSE) estimator
gives the best approximation (in the $L^2$-metric sense)
of the target nonlinear function in the dictionary subspace
(Proposition \ref{prp3} in Section \ref{subsec:whyhow}).
We highlight the fact that the HYPASS algorithm
implicitly exploits the reproducing kernel of the dictionary
subspace for updating the estimates (see Section \ref{subsec:hypass}).
We then show the way of constructing 
the reproducing kernel of a finite-dimensional subspace 
in terms of the Gram matrix of its basis
(Proposition \ref{prp1}  in Section \ref{subsec:whyhow}).
We can thus extend the strategy of HYPASS to any space
(which possibly has no reproducing kernel) in principle
as long as the Gram matrix is computable at least approximately.

The key idea is the following: 
(i) we make the function values well-defined in the dictionary subspace
by not considering the equivalence class, and
(ii) we then define the reproducing kernel of the dictionary subspace of the $L^2$ space.
For implementing the proposed method efficiently, we present three practical examples of computing the Gram matrix.
1) When the basis contains multiple Gaussian functions with different centers and scale parameters,
the inner product can analytically be computed by assuming that the input vector obeys
the normal distribution, or perhaps the improper constant distribution 
in analogy with a conjugate prior and a noninformative prior in Bayesian statistics.
2) The Gram matrix can be approximated with the atoms of the dictionary under a certain condition.
3) The Gram matrix can recursively be updated by using the matrix inversion lemma for rank-2 update.
We show that the approximate linear dependency (ALD) condition
\cite{engel04}
ensures a lower bound of the amount of the MMSE reduction
due to the newly entering dictionary-element,
keeping in mind  
the link between ALD and the coherence condition \cite{richard09}
(which we shall use for computational efficiency).
See Lemma \ref{lemma:coherence_ald} and Proposition \ref{prp4} in Section \ref{subsec:novelty}.
The computational complexity of the proposed algorithm
has the same order as that of the Euclidean approach
when the selective-update strategy is employed (see Section \ref{subsec:selective_update}).
Monotone approximation, asymptotic optimization, and convergence
of the proposed algorithm are proved for the full-updating case within the framework of 
the variable-metric adaptive projected subgradient method (APSM) \cite{APSM1,variable1}
(Theorem \ref{theoremconv} in Section \ref{sec:analysis}).
Numerical examples show that (i) the proposed algorithm enjoys
a better decorrelation property than CHYPASS \cite{yukawa_tsp15} and the multikernel NLMS (MKNLMS) algorithm \cite{yukawa_tsp12}, and (ii)
it outperforms, under the use of the selective-update strategy, the extended Kalman filter (EKF) for real data
as well as $13$ (out of $15$) batch learning methods that have been
compared in the literature \cite{elecpower1,elecpower2}.

\subsection{Relations to Bayesian and Stochastic Gradient Descent Approaches}
The projection-based methods tend to show better tracking/convergence with low computational complexity compared to the Bayesian and stochastic gradient descent approaches.
	By using the well-known {\it kernel trick}, the rigorous framework of the projection-based linear adaptive filtering has been extended to kernel adaptive filtering \cite{thslya_IEEESPMAG11,tayu_tsp15,tayu_tsp16,yukawa_tsp12}.
	Monotone approximation is one of the most significant properties of the projection-based methods,
	ensuring stable tracking when the target function keeps changing.  Convergence is also guaranteed when the target function is time-independent (see Theorem \ref{theoremconv} in Section \ref{sec:analysis} and its corresponding remark).
	Moreover, by virtue of the well-established algebraic properties of
	nonexpansive mappings {\cite[Chapter~17]{bauschke2011fixed}}, the projection-based methods have high flexibility of the algorithm design,
from the parallel-projection \cite{tayu_tsp15} and the multi-domain adaptive learning \cite{yukawaIEICE2010MultiDomain} to the sparsity-aware algorithms
\cite{tayu_tsp16,yukawa_tsp12}.

Those variants of the projection-based methods also lead to convergence speed comparable to the
	Bayesian approaches despite their low computational complexities.  Compared to the stochastic gradient descent algorithms such as NORMA \cite{kivinen04},
	(i) the projection-based methods offer tracking/convergence guarantees without elaborate step-size tuning, and
	(ii) can efficiently update the estimate even when the dictionary does not grow \cite{hypass} (see Section \ref{subsec:hypass}).
	In addition to the practical advantages, stable tracking capabilities and convergence guarantees for the variants
	can immediately be analyzed, as witnessed by the present work itself.
	Comparisons of the projection-based methods to Bayesian approaches (online Gaussian processes (GPs) \cite{csato01} and the kernel recursive least squares tracker (KRLS-T) \cite{vaerenbergh12}) and a stochastic gradient descent algorithm (NORMA) are summarized in \reftab{compmethods}.

\begin{table}[t]
	\begin{center}
		\caption{Comparisons of the projection-based methods to Bayesian and stochastic gradient descent approaches}
		\label{compmethods}
		\begin{tabular}{|c|c|c|c|c|}\hline
			Algorithm & Convergence & Tracking & Complexity & Variance \\
			&speed&&&information \\\hline\hline
			Online GPs &very fast&slow &high&yes \\\hline
			KRLS-T &fast&fast&high&yes \\\hline\hline
			NORMA &moderate&moderate&low&uninvestigated \\\hline\hline
			Projection &fast&very fast&low&uninvestigated \\\hline
		\end{tabular}
	\end{center}
\end{table}

\section{Preliminaries}
\label{sec:preliminary}
We first present the nonlinear system model under study
together with notation.
We then present our nonlinear estimator and 
its particular example, multikernel adaptive filtering model.
We finally review the HYPASS algorithm from another angle
based on a theorem on the reproducing kernel of a closed subspace.

\subsection{Nonlinear System Model}
\label{subsec:sysmodel}
Throughout, $\Real$, $\Natural$, and $\Naturalp$ are
the sets of real numbers, nonnegative integers, and positive integers,
respectively.
 We consider the following nonlinear system model:
 \begin{equation}
 	d_n:= 	\psi(\sigu_n)+\nu_n.  \label{problem1}
 \end{equation}
Here, the input (sample) vector $\sigu_n\in\real^L$ is assumed to be a random vector
with probability density function $p(\sigu)$,
$\nu_n$ is the additive noise at time $n\in\Natural$, and 
the nonlinear function $\psi$ is assumed to lie in
the real Hilbert space
 $\euclidspace:=L^2(\Real^L,\di\mu):=\{f \mid \norm{f}_{\euclidspace}<\infty\}$ 
equipped with the inner product
\begin{equation}
	\inpro{f,g}_{\euclidspace}:=
\Ex{f(\sigu)g(\sigu)}:=
\int_{\Real^L}f(\sigu)g(\sigu)\di\mu(\sigu), ~f,g\in \euclidspace, \label{eq:metric}
\end{equation}
and its induced norm
$\norm{f}_{\euclidspace}:=\sqrt{\inpro{f,f}_{\euclidspace}}$,
where $\di\mu(\sigu):=p(\sigu)\di\sigu$
is the probability measure.
Assuming that there exists $M\in(0,\infty)$ such that $p(\signal{u})<M$ for all $\signal{u}\in\real^L$, we have
\begin{equation}
\int|f(\signal{u})|^2p(\signal{u})\di \signal{u} \leq M \int|f(\signal{u})|^2\di \signal{u},
\end{equation}
which implies that $L^2(\Real^L,\di \signal{u})\subset L^2(\Real^L, \di \mu)=\euclidspace$.
It is known that the space $L^2(\Real^L,\di \signal{u})$ contains
any Gaussian RKHS as its subset \cite{smola98}.
Hence, our assumption $\psi\in\H$ is weaker than usually supposed in the literature of 
kernel adaptive filtering.

\noindent{\bf Notation:}
We denote by $\theta$ the null vector of $\H$.
The metric projection of a point $f\in\euclidspace$
onto a given closed convex set $C\subset \euclidspace$
is defined by
 \begin{equation}
 	P_{C}(f):=\argmin_{g\in C}\norm{f-g}_{\H}.
 \end{equation}
 If, in particular, $C$ is a linear variety (a translation of a linear
 subspace),
$P_{C}(f)$ is said to be the orthogonal projection.
Given $m$-dimensional real vectors $\sigx,\sigy\in\real^m$, define $\inpro{\sigx,\sigy}_{\Real^m}:=\sigx^{\T}\sigy$
and $\norm{\sigx}_{\Real^m}:=\sqrt{\inpro{\sigx,\sigx}_{\Real^m}}$, where $(\cdot)^{\T}$ stands for transposition.
Given any pair of integers $m,n\in\Natural$ such that
$m\leq n$, we denote by $\overline{m,n}$
the set of integers between $m$ and $n$; i.e.,
$\overline{m,n}:=\{m,m+1,\cdots,n\}$.
We denote the identity matrix by $\sigI$.

\subsection{Nonlinear Estimator}

Our nonlinear estimator takes the following form:
\begin{equation}
 \varphi_n:=\sum_{i=1}^{r_n} h_{i,n} f_i^{(n)}\in\mathcal{M}_n:={\rm
  span} \ \mathcal{D}_n, ~~~n\in\Natural,
\end{equation}
where
$\mathcal{D}_n:=
\{f_1^{(n)},f_2^{(n)},\cdots,f_{r_n}^{(n)}\}\subset \euclidspace$
is the dictionary at time $n$.
We assume that
the value $f_i^{(n)}(\signal{x})$ of each $f_i^{(n)}$,
$i\in\overline{1,r_n}$, at an arbitrary point $\signal{x}\in\real^L$
is predefined, i.e, $f_i^{(n)}$ is a representative of an equivalence class of functions in $\H$.
As will be seen in Section \ref{sec:contribution},
any set of functions in $\euclidspace$ can be used as a dictionary
in the proposed learning paradigm, as long as
$\innerprod{f_i^{(n)}}{f_j^{(n)}}_{\euclidspace}$,
$i,j\in\overline{1,r_n}$,
can be computed (or approximated) efficiently.
The evaluation of $\varphi_n$ at the current input $\sigu_n$ can be expressed as
\begin{equation}
	\varphi_n(\sigu_n)=\signal{f}_n(\sigu_n)^{\T}\sigh_n, \label{estout}
\end{equation}
where
$\sigh_n:=[h_{1,n},h_{2,n},\cdots,h_{r_n,n}]^{\T}\in\Real^{r_n}$,
and
$\signal{f}_n(\signal{u}):=
\left[f_1^{(n)}(\signal{u}), f_2^{(n)}(\signal{u}),
\cdots,
f_{r_n}^{(n)}(\signal{u})\right]^{\T}\in\Real^{r_n}$ for any
$\signal{u}\in\real^L$.

\subsection{Multikernel Adaptive Filtering Model}
\label{subsec:review}

We present a specific example of the dictionary $\mathcal{D}_n$.
Let $\mathcal{H}_1$, $\mathcal{H}_2$, $\cdots$,
$\mathcal{H}_Q$ be RKHSs equipped with the inner product
$\innerprod{\cdot}{\cdot}_{\euclidspace_q}$
and its induced norm $\norm{\cdot}_{\euclidspace_q}$,
$q\in\overline{1,Q}$.
Let $\kappa_q:\Real^L\times \Real^L\rightarrow \real$,
$q\in\overline{1,Q}$,
be the reproducing kernel of $\mathcal{H}_q$.
One of the celebrated examples is the Gaussian kernel
$\kappa_q(\signal{u},\signal{v}):=
\dfrac{1}{(2\pi\sigma_q^2)^{L/2}}
\exp\left(-\dfrac{\norm{\signal{u} - \signal{v}}_{\real^L}^2}
{2\sigma_q^2}\right)$, $\signal{u},\signal{v}\in\real^L$,
where 
$\sigma_q>0$ is the scale parameter with
$\sigma_1>\sigma_2>\cdots>\sigma_Q$.
The existing kernel/multikernel adaptive filtering approaches
exploit the properties of reproducing kernels:
(i) $\kappa_q(\cdot,\signal{u})\in \H_q$ and
(ii)
$f(\signal{u})=\innerprod{f}{\kappa_q(\cdot,\signal{u})}_{\euclidspace_q}$
$f \in \H_q$, $\signal{u}\in\real^L$.
We emphasize here that, given a space, 
different inner products give different reproducing kernels.
This is important to follow the discussions presented in Section
\ref{sec:contribution}.
We assume that $\mathcal{H}_q\subset \euclidspace$;
e.g., this assumption holds in the case of Gaussian kernels (see Section
\ref{subsec:sysmodel}).
For each $q\in\overline{1,Q}$ and each time instant $n\in\Natural$,
let $\mathcal{D}_{n}^{(q)}:=\{\kappa_q(\cdot,\signal{u}_j)\}_{j\in\mathcal{J}_n^{(q)}}$,
$\J_n^{(q)}:=\left\{j_{1,n}^{(q)},j_{2,n}^{(q)},...,j_{r_n^{(q)},n}^{(q)}\right\}\subset\overline{0,n}$ be the $q$th dictionary of size
 $r_n^{(q)}\in\Naturalp$,
$q\in\overline{1,Q}$.
The whole dictionary $\mathcal{D}_{n}:= \bigcup_{q\in\overline{1,Q}}\mathcal{D}_{n}^{(q)}$
at time $n$
is of size $r_n:=\sum_{q\in\overline{1,Q}}r_n^{(q)}$.

\subsection{HYPASS Algorithm Revisited: A Fresh View}
\label{subsec:hypass}

We start with the following theorem to find the reproducing kernel of a closed subspace of RKHS.

\begin{theorem}
[{\cite[Theorem~11]{berlinet2011reproducing}}]
	\label{theorem:rk_subspace}
	Let $\M$ be a closed subspace
of an RKHS $(\X,\innerprod{\cdot}{\cdot}_{\X})$ associated with the reproducing kernel $\kappa$. 
Then, $(\M,\innerprod{\cdot}{\cdot}_{\X})$ is an RKHS associated with
 the reproducing kernel $\kappa_{\M}$ given by
\begin{equation}
\kappa_{\M}(\sigu, \sigv) = P_{\M}^{\X}(\ka{\cdot,\sigv})(\sigu),
\end{equation}
where $P_{\M}^{\X}$ denotes the projection operator defined with respect
to the metric of $\X$.
(Note that $\kappa$ is not necessarily the reproducing kernel of
 $(\M,\innerprod{\cdot}{\cdot}_{\X})$,
because $\kappa(\cdot,\signal{u}) \not\in\M$ in general.)
\end{theorem}	

We consider the monokernel case of $Q=1$.
Taking a fresh look at 
the update equation of the HYPASS algorithm
\cite{hypass,tayu_tsp15}
under the light of Theorem \ref{theorem:rk_subspace},
we obtain
\begin{align}
	&\varphi_{n+1}:= 
	\varphi_n+\lambda_n(P_{\Pi_n}(\varphi_n)-\varphi_n)\\
	&=\varphi_n+\lambda_n
	\frac{d_n-\innerprod{\varphi_n}{\kappa_{1,\mathcal{M}_n}(\cdot,\sigu_n)}_{\euclidspace_1}}
{\norm{\kappa_{1,\mathcal{M}_{n}}(\cdot,\sigu_n)}_{\H_1}^2}\kappa_{1,\mathcal{M}_{n}}(\cdot,\sigu_n), 
\label{eq:hypass_freshview}
\end{align}
where $\lambda_n\in(0,2)$ is the step size,
$\kappa_{1,\mathcal{M}_{n}}$ is the reproducing kernel of $(\mathcal{M}_{n},\innerprod{\cdot}{\cdot}_{\euclidspace_1})$, and
\begin{equation}
\Pi_n:=\left\{f\in\mathcal{M}_n\mid f(\signal{u}_n)=\innerprod{f}{\kappa_{1,\mathcal{M}_n}(\cdot,\signal{u}_n)}_{\euclidspace_1}=d_n\right\}.
\label{eq:hyperplane}
\end{equation}
Here, the orthogonal decomposition \cite{luenberger} indicates that
$f(\signal{u})=\innerprod{f}{\kappa_1(\cdot,\signal{u})}_{\euclidspace_1}=\innerprod{f}{\kappa_{1,\M_n}(\cdot,\signal{u})}_{\euclidspace_1}$ for any $f\in\mathcal{M}_n$, and HYPASS can be regarded as projecting
the current estimate $\varphi_n$ onto the hyperplane $\Pi_n$
in the RKHS $(\M_n,\innerprod{\cdot}{\cdot}_{\H_1})$ of which the
reproducing kernel is $\kappa_{1,\mathcal{M}_n}$.
Note here that, if $\kappa_1(\cdot,\signal{u}_n)\in \M_n$, 
then it holds that
$\kappa_1(\cdot,\signal{u}_n) =
\kappa_{1,\M_n}(\cdot,\signal{u}_n)$.
\section{Proposed Online Learning Method}
\label{sec:contribution}
We first clarify why to use the $L^2$ metric for online learning, and
how to implement it.
We then present how to compute/approximate the autocorrelation matrix efficiently, and
explain the online dictionary-construction technique.
We finally discuss the complexity issue together with the selective-update strategy for complexity reduction.
\subsection{Online Nonlinear Estimation with $L^2$ Metric: Why \& How?}
\label{subsec:whyhow}

Given a dictionary, the MSE of a coefficient vector
$\sigh\in\real^{r_n}$ is given as
\begin{equation}
	\Ex{(d_n-\signal{f}_n(\sigu_n)^{\T}\sigh)^2}
	=\sigh^{\T}\sigR\sigh-2\sigh^{\T}\sigp+\Ex{d_n^2}, \label{cost1}
\end{equation}
where $\sigR:=E\left[\signal{f}_n(\sigu_n)\signal{f}_n(\sigu_n)^{\T}\right]\in\Real^{r_n\x r_n}$ is
the autocorrelation matrix of $\signal{f}_n(\sigu_n)$,
$\sigp:=E\left[\signal{f}_n(\sigu_n)d_n\right]\in\Real^{r_n}$
is the cross-correlation vector between $\signal{f}_n(\sigu_n)$ and $d_n$,
and $\Ex{\cdot}$ is the expectation taken over the input $\sigu$ as defined in \refeq{eq:metric}.
The following fact holds by definition of the inner product in \refeq{eq:metric}.
\begin{fact}
	\label{fact:G}
The autocorrelation matrix $\signal{R}$ is the Gram matrix of the
 dictionary $\mathcal{D}_n$ in 
the real Hilbert space $\H(:=L^2(\Real^L,\di\mu))$;
i.e., the $(i,j)$ entry of $\signal{R}$ is given by
$r_{i,j}:=\Ex{f_i^{(n)}(\signal{u})f_j^{(n)}(\signal{u})}=\innerprod{f_i^{(n)}}{f_j^{(n)}}_{\euclidspace}$.
\end{fact}

We can then show the following proposition.
\begin{proposition}
	\label{prp3}
	Assume that $E\left[\signal{f}_n(\sigu_n)\nu_n\right]=\signal{0}$.
Then, the MMSE estimator
	$\psi^{*}_{\M_n}:=\argmin_{f\in\M_n} E [d_n-$
	$f(\sigu_n)]^2$ coincides with 
the best approximation of $\psi$ in $\M_n$ in the $\euclidspace$-norm
 sense;
i.e., $\psi^{*}_{\M_n}=P_{\M_n}(\psi)$.
\end{proposition}
\begin{proof}
See Appendix \ref{appprp3}. 
\end{proof}
Proposition \ref{prp3} states that, in the proposed $L^2$ learning paradigm,
the MMSE estimator (what online algorithms tend to seek for) is the
best point in our actual search space $\mathcal{M}_n$.
This is in contrast to the existing kernel adaptive filtering paradigm \cite{yukawaletter16}
 (see Figure \ref{fig:tommse}).
 \begin{figure}[t]
 	\begin{center}
 		\psfrag{L}{\footnotesize{$\H:=L^2(\Real^L,\di\mu)$}}
 		\psfrag{M1}{\footnotesize{$\M_n$}}
 		\psfrag{opt}{\footnotesize{$\psi$}}
 		\psfrag{0}{\footnotesize{$\theta$}}
 		\psfrag{M22}{\hspace{-0.8em}\footnotesize{$P_{\M_n}(\psi)=\psi^{*}_{\M_n}$}}
 		\includegraphics[clip,width=0.36\textwidth]{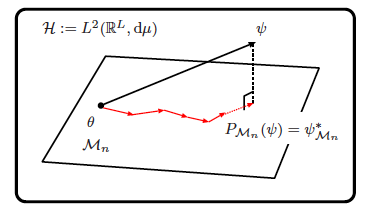}
 		\label{fig:proposed}
 		\centerline{(a) Proposed online learning paradigm}
 	\end{center}	
 	
 	\begin{center}
 		\psfrag{M1}{\footnotesize{$\M_n$}}
 		\psfrag{opt}{\footnotesize{$\psi$}}
 		\psfrag{0}{\footnotesize{$\theta$}}
 		\psfrag{M2}{}
 		\psfrag{M3}{\hspace{-0.8em}\footnotesize{$\psi^{*}_{\M_n}$}}
 		\psfrag{RKHS}{\footnotesize{${\rm RKHS}$}}
 		\includegraphics[clip,width=0.36\textwidth]{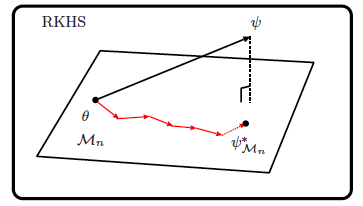}
 		\label{fig:rkhs}
 		\centerline{(b) Conventional kernel adaptive filtering paradigm}
 	\end{center}
 	\caption{The MMSE estimator $\psi^*_{\M_n}$ gives the best
 		approximation of the unknown function $\psi$
 		in the dictionary subspace $\mathcal{M}_n$ in the Hilbert space 
 		$(\H,\innerprod{\cdot}{\cdot}_{\euclidspace})$ (see Proposition
 		\ref{prp3}).
 		This is not generally true in an RKHS \cite{yukawaletter16}.
 	}
 	\label{fig:tommse}
 \end{figure}
When we consider online learning in the functional space $\H$, 
the MSE surface viewed in the Euclidean space 
(which is isomorphic to the dictionary subspace \cite{isomor1}) is determined
by the following function of the modified coefficient vector
$\hat{\sigh}:=\sigR^{\frac{1}{2}}\sigh$ \cite{tayu_tsp16}:
\begin{align}
	\Ex{(d_n-\hat{\signal{f}}_n^{\T}\hat{\sigh})^2}
	=\hat{\sigh}^{\T}\hat{\sigR}\hat{\sigh}-2\hat{\sigh}^{\T}\hat{\sigp}+\Ex{d_n^2},  \label{cost2}
\end{align}	
where $\hat{\signal{f}}_n:=\sigR^{-\frac{1}{2}}\signal{f}_n(\sigu_n)$,
 $\hat{\sigp}:=\sigR^{-\frac{1}{2}}\sigp$,
and $\hat{\sigR}:=\sigR^{-\frac{1}{2}}\sigR\sigR^{-\frac{1}{2}}=\signal{I}$.
This means that perfect decorrelation is achieved by adopting the $L^2$
 metric
$\innerprod{\cdot}{\cdot}_{\euclidspace}$, i.e, $\hat{\sigR}$ is the identity matrix in theory.
In other words,
$\euclidspace$ possesses the best geometry
for online nonlinear estimation in the sense of decorrelation under the possibly expanding dictionary subspace.
This is the core motivation of the present study.
Note that it is well known that a better-conditioned correlation matrix leads to faster convergence for {\it linear} adaptive filter (see \cite{haykin}, for example).

The question now is how to formulate a projection-based online learning
algorithm working in $\H$.
We have seen in Section \ref{subsec:hypass} that the normal vector
$\kappa_{\mathcal{M}_n}(\cdot,\signal{u}_n)$ of the hyperplane $\Pi_n$
gives the direction of update in \refeq{eq:hypass_freshview}, and 
it is readily available 
if the reproducing kernel $\kappa_{\mathcal{M}_n}$ is known.
As widely known, the $L^2$ space $\H$ has no reproducing kernel because
the value $f(\signal{u})$ of $f\in\H$ at a given point $\signal{u}\in\real^L$ is
not well defined due to the presence of equivalence classes
(i.e., those functions which coincide except for a measure-zero set are
regarded to be the same point).
Fortunately, however, what we need
is the reproducing kernel of the
dictionary subspace $\mathcal{M}_n$, as already mentioned.
In fact, if one regards $\varphi_n$ as an element of
$\euclidspace$,
its value $\varphi_n(\signal{u})$ at some specific point
$\signal{u}\in\real^{L}$ is not well defined.
Nevertheless, we define it by
$\varphi_n(\signal{u}) :=
\sum_{i=1}^{r_n} h_{i,n} f_i^{(n)}(\signal{u})$
as the value $f_i^{(n)}(\signal{u})$ is assumed to be predefined.
By doing so, $(\mathcal{M}_n,\innerprod{\cdot}{\cdot}_{\H})$ becomes
a finite-dimensional real Hilbert space in which
the value of each function at each point is well defined.
In this case, there is a systematic way to construct 
the reproducing kernel of the space, as shown below.

\begin{proposition}
	\label{prp1}
	Let $\Dict:=\{f_1,f_2,...,f_r\}\subset \H$, $r\in\Naturalp$, be
	an independent set, and
	$\sigG$ the Gram matrix with its $(k,l)$ entry
	$g_{k,l}:=\inpro{f_k,f_l}_{\H}$.
	Define
	$\signal{f}(\sigu):=\left[f_1(\sigu),f_2(\sigu),\cdots,f_r(\sigu)\right]^{\T}$,
	$\sigu\in\Real^L$.
	Then,
	\begin{equation}
		\ka{\sigu,\sigv}:=\signal{f}(\sigu)^{\T}\sigG^{-1}\signal{f}(\sigv),\; \sigu,\sigv\in\Real^L, \label{repro}
	\end{equation}
	is the reproducing kernel of the Hilbert space $\left({\rm span}\Dict,\inpro{\cdot,\cdot}_{\H}\right)$.
\end{proposition}
\begin{proof}
 See Appendix \ref{appprp1}.
(The case when $\Dict$ is an orthonormal set can be found in \cite[page 7 - Example~1]{berlinet2011reproducing}.)
\end{proof}

The reproducing kernel of the RKHS
$(\mathcal{M}_n,\innerprod{\cdot}{\cdot}_{\H})$
is then given by
	\begin{equation}
		\kappa_{\mathcal{M}_n}(\sigu,\sigv):=
\signal{f}_n^{\T}(\sigu)\sigR^{-1}\signal{f}_n(\sigv),~~~
\sigu,\sigv\in\Real^L.
\label{eq:rk_Mn}
	\end{equation}
We mention for clarity that
\begin{equation}
 f(\signal{u})=\innerprod{f}{\kappa_{1,\mathcal{M}_n}(\cdot,\signal{u})}_{\euclidspace_1}=\innerprod{f}{\kappa_{\mathcal{M}_n}(\cdot,\signal{u})}_{\euclidspace}.
\end{equation}
Let us remind here that the reproducing kernel depends on the inner
product, and $\kappa_{1,\mathcal{M}_n}$
is the reproducing kernel of
($\mathcal{M}_n,\innerprod{\cdot}{\cdot}_{\mathcal{H}_1}$) while
$\kappa_{\mathcal{M}_n}$
is the reproducing kernel of
($\mathcal{M}_n,\innerprod{\cdot}{\cdot}_{\mathcal{H}}$).

We are now ready to present the proposed algorithm.
Define
the bounded-instantaneous-error hyperslab
\begin{align}
C_n:=\left\{f\in\mathcal{M}_n\mid |f(\signal{u}_n)-d_n|\leq\rho\right\}, \label{pro2}
\end{align}
where
$\rho\geq0$.
For the initial estimate $\varphi_0:=\theta$,
generate the sequence $(\varphi_n)_{n\in\Natural}$ of nonlinear estimators by
\begin{align}
\hspace*{-2em}&	\varphi_{n+1}:=\varphi_n + \lambda_n\left(
	P_{C_n}(\varphi_n) - \varphi_n
	\right)\nonumber\\
\hspace*{-2em}	&	=
	\varphi_n+\lambda_n\sgn{\left(e_n(\sigu_n)\right)}
	\frac{\max{\{|e_n(\sigu_n)|-\rho,0\}}}{\norm{\kappa_{\mathcal{M}_n}(\cdot,\sigu_n)}_{\H}^2}
\kappa_{\mathcal{M}_n}(\cdot,\sigu_n), \label{pro1}
\end{align}
where $\sgn{(\cdot)}$ is the sign function, $e_n(\signal{u}_n)\!:=\!d_n - \varphi_n(\signal{u}_n)$, and
$\lambda_n\in(0,2)$ is the step size.

We have shown how to update our nonlinear estimator $\varphi_n$, given a
dictionary $\mathcal{D}_n$.
The remaining issues to be discussed are
how to compute $\sigR$ in \refeq{eq:rk_Mn} efficiently (Section
\ref{subsec:R})
and 
how to construct the dictionary $\Dict_n$ (Section \ref{subsec:novelty}).
We shall also present the selective-update strategy 
to reduce the computational complexity in Section \ref{subsec:selective_update}.


\subsection{Practical examples of computing $\sigR$ efficiently}
\label{subsec:R}

We present three options to estimate/approximate $\sigR$ efficiently.
The first option assumes the use of multiple Gaussian functions with different scales
(see Section\ref{subsec:review}), while
the other two options can be applied to the general case.

\subsubsection{Analytical approach}
\label{subsubsec:analytic}
We present two examples in which analytical expressions of inner product can be obtained
by using the analogy to a conjugate prior and a noninformative prior \cite{bishop06}.
\begin{proposition}
	\label{prp2}
	Let $\kappa_p(\cdot,\sigu),~p\in\overline{1,Q},~\signal{u}\in\real^L$, and 
$\kappa_q(\cdot,\sigv),~ q\in\overline{1,Q},~\signal{v}\in\real^L$,
	be two Gaussian functions with scale parameters $\s_p,\s_q>0$, respectively.\\
		
	\noindent	(a) Case of Gaussian input:\\
Assume that the input vector $\sigu\in\Real^L$ follows the normal distribution with variance $\s^2$, i.e., the probability density function for the input vector is given by
\begin{equation}
	p(\sigu):=\frac{1}{(2\pi\s^2)^{L/2}}\exp{\left(-\frac{\norm{\sigu}^2_{\Real^L}}{2\s^2}\right)},\;\;\;\s>0. \label{gauss}
\end{equation}
Then, the inner product can be given analytically by
\begin{align}
\hspace*{-2em}	&\inpro{\kappa_p(\cdot,\sigu),\kappa_q(\cdot,\sigv)}_{\H}
	=\nn\\
\hspace*{-2em}	&\frac{1}{(2\pi)^L\upsilon^{L/2}} \exp\left(-\frac{\s^2\norm{\sigu-\sigv}_{\Real^L}^2+\s_q^2\norm{\sigu}_{\Real^L}^2+\s_p^2\norm{\sigv}_{\Real^L}^2}{2\upsilon}\right), \label{case1}
\end{align}	
where $\upsilon:=\s^2\s_p^2+\s^2\s_q^2+\s_p^2\s_q^2$.\\
	(b) Case of unknown input distribution: \\
Suppose that there is no available information about the input distribution.  In this case, 
by using the analogy to a noninformative prior which is improper, let $\di\mu(\sigu):=\di\sigu$, i.e, the input is assumed to distribute uniformly over the infinite interval.
The inner product is then given by
\begin{align}
	&\inpro{\kappa_p(\cdot,\sigu),\kappa_q(\cdot,\sigv)}_{\H} \nonumber\\
	&=\frac{1}{(2\pi(\s_p^2+\s_q^2))^{L/2}}
	\exp\left(-\frac{\norm{\sigu-\sigv}_{\Real^L}^2}{2(\s_p^2+\s_q^2)}\right). \label{inprocompute}
\end{align}	
\end{proposition}
\begin{proof}
 See Appendix \ref{appprp2}.
\end{proof}

\subsubsection{Finite-sample approach --- use of sample average}
\label{subsubsec:dictapp}
It is also possible to approximate $\sigR$ by a sample average.
	Let $(\sigu_j)_{j\in\J_n}$, where $\J_n:=\{j_1,j_2,...,j_{l_n}\}\subset\overline{0,n}$, be a fixed set of realizations of the input vectors $\sigu_n$.
	Then, at time $n$, the matrix $\sigR$ is approximated by
	\begin{align}
	\sigR\approx\frac{1}{l_n}\signal{F}_n\signal{F}_n^{\T}, \label{dictapprox}
	\end{align}
	where $l_n\in\Naturalp$ is the size of $\J_n$, and $
		\signal{F}_n:=\left[
		\begin{array}{cccc}
		\signal{f}_n(\sigu_{j_1})&\signal{f}_n(\sigu_{j_2})&\cdots&\signal{f}_n(\sigu_{j_{l_n}})
		\end{array}
		\right]$.
When the dictionary elements are associated with input vectors,
the set $(\sigu_j)_{j\in\J_n}$ might be given as the set of dictionary data (e.g. $\J_n=\bigcup_{q\in\overline{1,Q}}\J_n^{(q)}$ for the case of multikernel adaptive filtering).
Suppose that $\H$ is an RKHS with $f_i^{(n)}:=\ka{\cdot,\sigu_{j_i}},\;j_i\in\J_n\;(l_n=r_n)$, where $\kappa$ is supposed to be the reproducing kernel of $\H$.   Then, $\signal{F}_n$ is the Gram matrix of the dictionary $\Dict_n$.  Hence, the approximation in \refeq{dictapprox} is a natural extension of the
$G^2$-metric studied in \cite{takeuchi16}.
\subsubsection{Recursive approach}
\label{subsubsec:variablemetric}
The inverse autocorrelation matrix $\signal{R}^{-1}$ appearing in \refeq{eq:rk_Mn} can be approximated,
recursively, by using a similar trick to the kernel recursive least squares (KRLS) algorithm \cite{engel04}.
We assume that the dictionary may only change in an incremental way;
i.e., $\mathcal{D}_{n-1}\subseteq\mathcal{D}_{n}$ and $r_n\in\{r_{n-1}, r_{n-1}+1\}$.
(It is straightforward to scale down the size of the autocorrelation matrix
when some elements are excluded from the dictionary.)
When the dictionary is unchanged, the estimate of the autocorrelation matrix can be updated as
$\sigR_n:=\sigR_{n-1}+\sigf_n\sigf_n^{\T}$, and its inverse
$\sigR_n^{-1}$ can be updated recursively as
\begin{equation}
	\sigR_n^{-1}=\sigR_{n-1}^{-1}-\frac{\sigR_{n-1}^{-1}\sigf_n\sigf_n^{\T}\sigR_{n-1}^{-1}}{1+\sigf_n^{\T}\sigR_{n-1}^{-1}\sigf_n}. \label{gram1}
\end{equation}	
When a new basis function is added and the dictionary is changed,
we define the estimate of the augmented autocorrelation matrix as
\begin{align}
	&\sigR_n:=
\left[
	\begin{array}{cc}
	\sigR_{n-1} & \signal{0}\\
	\signal{0}^{\T} & 0
	\end{array}
	\right] + \signal{f}_n(\sigu_n)\signal{f}_n(\sigu_n)^{\sf T} =
\sigA+\signal{B}\signal{C}, \label{append}
\end{align}
where
$\sigA:=\left[
	\begin{array}{cc}
	\sigR_{n-1} & \signal{0}\\
	\signal{0}^{\T} & 1
	\end{array}
	\right],~
	\signal{B}:=\left[
	\begin{array}{cc}
	\sigf_n & \signal{e}_{r_n}
	\end{array}
	\right],~\signal{C}:=\left[
	\begin{array}{cc}
	\sigf_n & -\signal{e}_{r_n}
	\end{array}
	\right]^{\T}$,
and $\signal{e}_{r_n}:=[0,0,\cdots,0,1]^{\T}\in\Real^{r_n}$.
Assuming that $f_{r_n}^{(n)}(\sigu_n)\neq0$ to ensure the nonsingularity of
$\sigI+\signal{C}\sigA^{-1}\signal{B}=\left[
\begin{array}{cc}
*& f_{r_n}^{(n)}(\sigu_n)\\
-f_{r_n}^{(n)}(\sigu_n) & 0\\
\end{array}
\right]$,
one can apply the matrix inversion lemma to compute the inverse of
the rank-2 update \refeq{append},
obtaining the following recursion:
\begin{align}
	\sigR_n^{-1}=\sigA^{-1}-\sigA^{-1}\signal{B}(\sigI+\signal{C}\sigA^{-1}\signal{B})^{-1}\signal{C}\sigA^{-1}.
\end{align}
The idea of this ``Recursive approach'' comes certainly from 
the recursive least squares (RLS) algorithm,
which iteratively minimizes the sum of the squared errors.
In fact, RLS can be viewed as a variable-metric projection algorithm
with nearly-unity step size \cite{yukawalecture} (see Appendix \ref{subsec:recursiveapp}).
\subsection{Dictionary Construction with Novelty Criterion}
\label{subsec:novelty}

The dictionary is constructed based on some novelty criterion as
follows:
a function $f_{\signal{u}_n}\in\H$ depending on the new measurement $\signal{u}_n$
is added into the dictionary if
it satisfies some prespecified novelty criterion.
In the particular case of multiple Gaussian functions 
(see Section\ref{subsec:review}), a possible option is the following.
\begin{enumerate}
	\item The coarsest Gaussian function $\kappa_1(\cdot,\sigu_n)$ is added into 
	the dictionary when it satisfies the novelty criterion.
	\item A finer Gaussian $\kappa_i(\cdot,\sigu_n)$, $i\geq 2$, is added into the dictionary
	if it satisfies the novelty criterion
	but all the coarser Gaussians $\kappa_1(\cdot,\sigu_n)$, $\kappa_2(\cdot,\sigu_n)$,
	$\cdots$, $\kappa_{i-1}(\cdot,\sigu_n)$, do not.
\end{enumerate}

In analogy with Platt's criterion \cite{platt91},
we consider two novelty conditions both of which need to be satisfied:
(i) the coherence condition (elaborated below) and
(ii) the large-normalized-error (LNE) condition
\begin{equation}
|d_n-\varphi_n(\sigu_n)|^2= \abs{e_n(\signal{u}_n)}^2 >
\epsilon|\varphi_n(\sigu_n)|^2,\; \epsilon\geq0.  \label{novel2}
\end{equation}
Given a threshold $\delta\in[0,1]$,
the coherence condition is given as follows:
\begin{align}
	\max_{f\in\mathcal{D}_n}
c(f,f_{\signal{u}_n})\leq\delta, \label{cohe}
\end{align}	
where $c(f,g):=\frac{\abs{\innerprod{f}{g}_{\H}}}{\norm{f}_{\H}\norm{g}_{\H}}$, $f,g\in\H$.
Note here that, by definition, those factors $\innerprod{f}{f_{\signal{u}_n}}_{\H}$, $\norm{f}_{\H}$, and
$\norm{f_{\signal{u}_n}}_{\H}$ involve expectation, which brings the same issue as
for the computation of $\signal{R}$ discussed in Section \ref{subsec:R}.
When Analytical approach is employed 
for the computation of $\signal{R}$, Proposition \ref{prp2} can be applied.
When Finite-sample/Recursive approach is employed, one may use sample averages
with the $s_n\in\Natural^*$ most-recent  measurements, for instance, as
\begin{align}
\innerprod{f}{f_{\signal{u}_n}}_{\H}\approx \frac{1}{s_n}\sum_{i=n-s_n+1}^{n}
f(\signal{u}_i) f_{\signal{u}_n}(\signal{u}_i). \label{arithave}
\end{align}

The following lemma links the coherence condition to the ALD condition.

\begin{lemma}
	\label{lemma:coherence_ald}
Assume that $(r_n-1)\delta<1$.
Then, the coherence condition \refeq{cohe} ensures 
the following ALD condition:
	\begin{equation}
		\frac{\norm{f_{\signal{u}_n}-P_{\M_n}(f_{\signal{u}_n})}^2_{\H}}{\norm{f_{\signal{u}_n}}^2_{\H}}\geq1-\frac{(r_n-1)\delta^2}{1-(r_n-2)\delta}.
\label{eq:coherence_ald}
	\end{equation}
\end{lemma}
\begin{proof}
The assertion can be verified by \cite[Equation (16)]{richard09}
with the simple observion that the left-hand side of \refeq{eq:coherence_ald} equals to
$\norm{\frac{f_{\signal{u}_n}}{\norm{f_{\signal{u}_n}}_{\H}}-P_{\M}\left(\frac{f_{\signal{u}_n}}{\norm{f_{\signal{u}_n}}_{\H}}\right)}^2_{\H}$.
\end{proof}

Due to the property $\psi^{*}_{\M_n}=P_{\M_n}(\psi)$ presented in
Proposition \ref{prp3}, the proposed online learning algorithm with the $L^2$ metric
takes a particular benefit from ALD, as indicated by the following proposition.
\begin{proposition}
	\label{prp4}
Suppose that $f_{\signal{u}_n} \in \mathcal{D}_{n+1}$,
i.e., $f_{r_{n+1}}^{(n+1)} := f_{\signal{u}_n}$, and  ${\rm dim} \ \mathcal{M}_{n+1}=r_{n+1}$.
	Assume that $E\left[\signal{f}_{n+1}\nu_n\right]=\signal{0}$, and 
$E\left[\psi(\sigu_n)\nu_n\right]=0$.
Assume also that the ALD condition
\begin{equation}
	\frac{\norm{f_{\signal{u}_n}-P_{\M_n}(f_{\signal{u}_n})}^2_{\H}}
{\norm{f_{\signal{u}_n}}^2_{\H}}\geq\eta	\label{ald1}
\end{equation}
is satisfied for a given threshold $\eta\in[0,1]$.
Then, for the MMSE estimators
$\psi^{*}_{\M_n}(:=\argmin_{f\in\M_n} E [d_n-$
	$f(\sigu_n)]^2)$ and
$\psi^{*}_{\M_{n+1}}$, it holds that
	\begin{align}
\Ex{d_n-\psi^{*}_{\M_{n}}(\sigu_n)}^2
- \Ex{d_n-\psi^{*}_{\M_{n+1}}(\sigu_n)}^2\nn\\
\hfill	\geq	(h_{{\sigu_n}}^*)^2\norm{f_{\sigu_n}}_{\H}^2 \eta,
 \label{ald2}
	\end{align}	
where $h_{{\sigu_n}}^*\in\real$ is the coefficient of $f_{\sigu_n}$ in the expansion of $\psi^{*}_{\M_{n+1}}$.
\end{proposition}
\begin{proof}
 See Appendix \ref{appprp4}.
\end{proof}

Proposition \ref{prp4} states that the amount of MMSE reduction is at least
$(h_{{\sigu_n}}^*)^2\norm{f_{\sigu_n}}_{\H}^2\eta$ under the ALD condition in the space $\H$.
The coherence condition actually ensures the ALD condition for $\eta=1-\frac{(r_n-1)\delta^2}{1-(r_n-2)\delta}$ as long as $(r_n-1)\delta<1$
(see Lemma \ref{lemma:coherence_ald}), thereby yielding efficient MMSE reduction.
When the condition $(r_n-1)\delta<1$ is violated, an alternative option
that could have a better performance-complexity tradeoff
is to select $s_n\in\Natural^*$ elements from $\mathcal{D}_n$ that are maximally
coherent to the $f_{\signal{u}_n}$ \cite{tayu_tsp15} and check the ALD condition
with respect to the selected elements.


\subsection{Complexity Reduction by Selective Update}
\label{subsec:selective_update}
Although matrix inversion requires cubic complexity in general,
the complexity is $\mathO(r_n^2)$ when Analytical approach (Section \ref{subsubsec:analytic}) or Recursive approach (Section \ref{subsubsec:variablemetric}) are adopted.
However, it is still computationally expensive when the dictionary size becomes large.
Therefore, in practice, one may use the selective-update strategy, i.e., select a subset
$\tilde{\mathcal{D}}_n(\subset \mathcal{D}_n)$
of cardinality $\abs{\tilde{\mathcal{D}}_n}=s_n\in\overline{1,r_n}$, such that
$c(f,\kappa_{\M_n}(\cdot,\sigu_n)) \geq 
c(g,\kappa_{\M_n}(\cdot,\sigu_n))$
for any $f\in\tilde{\mathcal{D}}_n$
and for any $g\in\mathcal{D}_n\setminus \tilde{\mathcal{D}}_n$, where
\begin{align}
& c(f,\kappa_{\M_n}(\cdot,\sigu_n))
=\frac{|f(\sigu_n)|}
{\norm{f}_{\H}\sqrt{\kappa_{\M_n}(\sigu_n,\sigu_n)}}. \label{coheselec}
\end{align}	
Let $\tilde{\mathcal{D}}_n:=\{f^{(n)}_1,f^{(n)}_2,\cdots,f^{(n)}_{s_n}\}$ without loss of generality.
The update equation is then given in a parametric form as
\begin{equation}
\tilde{\sigh}_{n+1}=\tilde{\sigh}_n+\lambda_n\sgn{\left(e_n(\signal{u}_n)\right)}\frac{\max{\{|e_n(\signal{u}_n)|-\rho,0\}}}{\tilde{\sigf}_n^{\T}\tilde{\sigR}_n^{-1}\tilde{\sigf}_n+\gamma_{\rm update}}\tilde{\sigR}_n^{-1}\tilde{\sigf}_n, \label{selecup}
\end{equation}
where $\gamma_{\rm update}\geq0$ is the regularization parameter, $\tilde{\sigh}_n:=[h_{1,n},h_{2,n},\cdots,h_{s_n,n}]^{\T}\in\Real^{s_n}$ is the coefficient vector
corresponding to the selected basis functions,
$\tilde{\sigf}_n:=\tilde{\sigf}_n(\sigu_n):=[f^{(n)}_1(\sigu_n),f^{(n)}_2(\sigu_n),\cdots,f^{(n)}_{s_n}(\sigu_n)]^{\T}\in\Real^{s_n}$, and
 $\tilde{\sigR}_n$ is the submatrix of $\sigR_n$ corresponding to the selected dictionary $\tilde{\mathcal{D}}_n$.

It is straightforward to obtain $\tilde{\sigR}_n$ by applying Proposition \ref{prp2} when Analytical approach is employed.
Otherwise, only the submatrix $\tilde{\sigR}_n$ of $\sigR_n$ is updated at time $n$ as
\begin{align}
	&\tilde{\sigR}_n:=\tilde{\sigR}_{n-1}+\tilde{\sigf}_n\tilde{\sigf}_n^{\T},  \label{selec1}
\end{align}
or, it is approximated by using Finite-sample approach as
\begin{align}
\tilde{\sigR}_n\approx\frac{1}{l_n}\tilde{\signal{F}}_n\tilde{\signal{F}}_n^{\T}, \label{selec2}
\end{align}
where $
\tilde{\signal{F}}_n:=\left[
\begin{array}{cccc}
\tilde{\sigf}_n(\sigu_{j_1})&\tilde{\sigf}_n(\sigu_{j_2})&\cdots&\tilde{\sigf}_n(\sigu_{j_{l_n}})
\end{array}
\right]
$.
The proposed online learning algorithm, including the selective-update strategy and
dictionary constructions, is summarized in Algorithm \ref{modelalgo}.
\begin{algorithm}[t!]
	\caption{Online Nonlinear Estimation via Iterative $L^2$-Space Projections}
	\label{modelalgo}
	\begin{algorithmic}
		\STATE {\bfseries Requirement:} $(\lambda_n)_{n\in\Natural}\subset[\epsilon_1,2-\epsilon_2]\subset(0,2)$, $\exists\epsilon_1,\epsilon_2>0$, 
		$\rho\geq0$ (hyperslab), 
		$\gamma\in(0,1)$ (regularization for $\signal{R}_n$), 
		$\gamma_{\rm update}\geq0$ (regularization for coefficient updates), 
		$\delta\in[0,1]$ (coherence),
		$\epsilon\geq0$ (LNE), and $s_n\in\overline{1,r_n}$ (efficiency factor)
		\STATE {\bfseries Initialization:} $\varphi_0:=\theta$, $\Dict_{-1}=\emptyset$
		\STATE {\bfseries Output:}
		$\varphi_n(\sigu_n):=\sum_{i=1}^{r_n} h_{i,n} f_i^{(n)}(\sigu_n)$
		\FOR{$n\in\Natural$}
		\STATE - Receive $\sigu_n\in\Real^L$ and $d_n\in\Real$
		\STATE - Check if the novelty criterion is satisfied for a candidate function $f_{\sigu_n}$ \hfill $\triangleright$ \refeq{novel2}, \refeq{cohe}\\
		\hspace*{2em} Coherence computation \hfill $\triangleright$ Proposition \ref{prp2} or \refeq{arithave}
		\IF{Novelty criterion is satisfied} 
		\STATE Dictionary increment:
		$\Dict_n=\Dict_{n-1}\cup\{f_{\sigu_n}\}$, $h_{r_n,n}=0$
		\ENDIF
		\STATE - Select $s_n$ coefficients to update
		\hfill $\triangleright$  \refeq{coheselec}
		\STATE - Compute $\tilde{\sigR}_n$
		\hfill $\triangleright$  Proposition \ref{prp2}, \refeq{selec1}, or \refeq{selec2}
		\STATE - Update $\tilde{\sigh}_n$
		\hfill $\triangleright$  \refeq{selecup}
		\ENDFOR
	\end{algorithmic}
\end{algorithm}
\begin{table}[t]
	\begin{center}
		\caption{Computational complexities of the algorithms}
		\label{alg}
		\begin{tabular}{|c|c|}\hline
			NLMS  & $3L+2$\\\hline
			KNLMS & $(L+6)r_n+2$\\\hline
			KRLS-T  & $5r_n^2+(L-5)r_n+1$\\\hline
			HYPASS  & $(L+4)r_n+\frac{L+5}{2}s_n^2-\frac{L-1}{2}s_n+2+v_{\rm inv}(s_n)$\\\hline
			MKNLMS & $(L+6)r_n+2$\\\hline
			CHYPASS &$(L+5)r_n+\frac{L+5}{2}s_n^2-\frac{L-1}{2}s_n+2+v_{\rm inv}(s_n)$ \\\hline
			Analytical &$(L+5)r_n+\frac{L+5}{2}s_n^2-\frac{L-1}{2}s_n+2+v_{\rm inv}(s_n)$ \\\hline
			Finite-sample &$\{L+10+(L+5)s_n+(L+4)\frac{s_n^2-s_n}{2}\}r_n$\\
			&$+s_n^2+2s_n+2+v_{\rm inv}(s_n)$ \\\hline
			Recursive &$\{L+11+(L+5)s_n\}r_n+\frac{L+6}{2}s_n^2-\frac{L}{2}s_n$ \\
			&$+2+v_{\rm inv}(s_n)$\\\hline
		\end{tabular}
	\end{center}
\end{table}
\begin{figure}[t]
	\begin{center}
		\psfrag{KNLMS,MKNLMS}{\footnotesize{KNLMS,MKNLMS}}
		\psfrag{HYPASS,CHYPASS,Analytical}{\hspace{-2.5em}\footnotesize{HYPASS, CHYPASS, Analytical}}
		\psfrag{NLMS}{\footnotesize{NLMS}}
		\psfrag{Recursive}{\footnotesize{Recursive}}
		\psfrag{Dictionary}{\footnotesize{Finite-sample}}
		\psfrag{Complexity}{\footnotesize{Complexity}}
		\psfrag{KRLS-T}{\footnotesize{KRLS-T}}
		\psfrag{Dictionary size}{\footnotesize{Dictionary size}}
		\includegraphics[clip,width=0.47\textwidth]{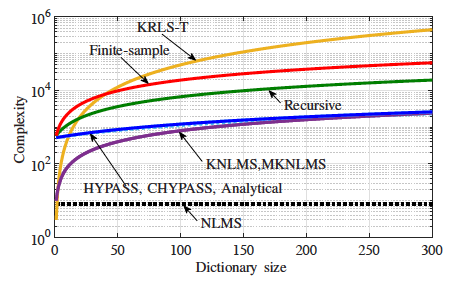}
		\caption{Evolutions of computational complexities of the algorithms for $L=2$ and $s_n=7$.  The proposed algorithm is of linear order to the dictionary size as implied in \reftab{alg}}
		\label{fig:evocomplexity}
	\end{center}
\end{figure}

\noindent{\bf Complexity:}
We discuss the computational complexity in terms of the number of multiplications required at each iteration
when the normalized Gaussian functions are used.
Suppose that the coherence condition is employed. 
The coherence condition only requires $\mathO(r_n)$ complexity whereas the ALD condition requires $\mathO(r_n^2)$ complexity.
Suppose also that we employ the selective-update strategy with the efficiency factor $s_n$, and that
\refeq{arithave} and \refeq{selec2} are used for $l_n:=s_n$.
\reftab{alg} summarizes the overall per-iteration complexity of the proposed algorithm,
NLMS \cite{NLMS1}, KNLMS \cite{richard09}, KRLS-T \cite{vaerenbergh12}, HYPASS \cite{tayu_tsp15}, MKNLMS, and CHYPASS.
Here, Analytical, Finite-sample, and Recursive in the table correspond respectively to
Analytical, Finite-sample, and Recursive approaches
presented in Section \ref{subsec:R}.
The complexity required for the inverse of an $s_n\x s_n$ matrix is denoted by $v_{\rm inv}(s_n)$ in \reftab{alg}.
\reffig{fig:evocomplexity} shows the evolutions of computational complexities of the algorithms for $L=2,s_n=7$, $n\in\Natural$ (we let $v_{\rm inv}(s_n):=s_n^3$). 

\section{Convergence Analysis}
\label{sec:analysis}
In this section, convergence analysis (together with monotone approximation) of the proposed algorithm is
presented for the full-updating case; i.e., the case of $s_n=r_n$.
 (Note here that the analysis for $s_n< r_n$ is intractable \cite{tayu_tsp15}).
Before presenting the analysis, we show how the proposed algorithm can be derived from APSM \cite{APSM1}.
Let $\Theta_n:\H\rightarrow[0,\infty),\;n\in\Natural$, be continuous convex functions
and $K\subset\H$ a nonempty closed convex subset.
For an arbitrary $\phi_0\in\H$, APSM \cite{APSM1} generates the sequence $(\phi_n)_{n\in\Natural}\subset K$ as
\begin{align}
\phi_{n+1}:=\begin{cases}
& P_K\left(\phi_n-\lambda_n\frac{\Theta_n(\phi_n)}{\norm{\Theta_n^{'}(\phi_n)}_{\H}^2}\Theta_n^{'}(\phi_n)\right),\\
&\;\;\;\;\;\;\;\;\;\;\;\;\;\;\;\;\;\;\;\;\;\;\;\;\;\;\;{\rm if}\;\Theta_n^{'}(\phi_n)\neq\theta,\\
& \phi_n,\;\;\;\;\;\;\;\;\;\;\;\;\;\;\;\;\;\;\;\;\;{\rm if}\;\Theta_n^{'}(\phi_n)=\theta,
\end{cases}	\label{APSM}
\end{align}	
where $\lambda_n\in[0,2],\;n\in\Natural$, and 
$\Theta_n^{'}(\phi_n)$ is a subgradient of $\Theta_n$ at $\phi_n$.
Letting
\begin{align}
\Theta_n(\varphi) :=\norm{\varphi-P_{C_n}(\varphi)}_{\H}
\end{align}
and $K:=\H$ in APSM reproduces the proposed algorithm.
More precisely, the metric of $\H$ is characterized by the autocorrelation matrix $\sigR$ in the dictionary subspace (cf. Fact \ref{fact:G}), and the proposed algorithm exploits the efficiently computable $\sigR_n$ in lieu of $\sigR$ (which is unavailable in practice).  This means that the metric used is fairly close to that of $\H$ but it involves time variations.
We therefore present our analysis based on the variable-metric version of APSM \cite{variable1}.
We first present a set of assumptions (see {\cite[Assumption~1]{tayu_tsp15}} and \cite[Assumption~2]{variable1}).
\begin{assumption}
	\label{assump1ana}
	\begin{enumerate}
		\item Step-size condition: there exist $\epsilon_1,\epsilon_2>0$ such that $(\lambda_n)_{n\in\Natural}\subset[\epsilon_1,2-\epsilon_2]\subset(0,2)$.
		\item Boundedness of dictionary size: there exists some $N_0\in\Natural$ such that $\M_n=\M_{N_0}$ for all $n\geq N_0$.
		\item Data consistency: there exists some $N_1\geq N_0$ such that 
$\Omega:=\bigcap_{\substack{n>N_1\\\varphi_n\notin C_n}} C_n$ has an interior point in the Hilbert space
		$(\M_{N_0},\inpro{\cdot,\cdot}_{\H})$, where $C_n$ is the bounded-instantaneous-error hyperslab defined in \refeq{pro2}.
		\item Boundedness of the eigenvalues of $\sigR_n$:
		there exist $\delta_{\rm min},\;\delta_{\rm max}\in(0,\infty)$ s.t.
		$\delta_{\rm min}<\s_{\sigR_n}^{\rm min}\leq\s_{\sigR_n}^{\rm max}<\delta_{\rm max}$ for all $n\in\Natural$,
		where $\s_{\sigR_n}^{\rm min}$ and $\s_{\sigR_n}^{\rm max}$ are the minimal and maximal eigenvalues of $\sigR_n$, respectively.
		\item Small metric-fluctuations:
		There exist some constant positive-definite matrix
$\sigP\in \Real^{r_{N_0}\x r_{N_0}}$, 
nonempty subset $\Gamma\subset \Omega$,
integer $N_2(\geq N_1)$, and positive constant $\tau>0$ s.t.
		$\signal{E}_n:=\sigR_n-\sigP\in \Real^{r_{N_0}\x r_{N_0}}$ satisfies
		\begin{align}
		&\frac{\norm{\sigh_{n+1}+\sigh_n-2\sigh^{*}}_{\real^{r_{N_0}}}\norm{\signal{E}_n}_2}{\norm{\sigh_{n+1}-\sigh_n}_{\real^{r_{N_0}}}} \nn\\
		&~~~~~~~~~~~~~~~~~~~<
		\frac{\epsilon_1\epsilon_2\s_{\sigP}^{\rm min}\delta_{\rm min}^2}{(2-\epsilon_2)^2\s_{\sigP}^{\rm max}\delta_{\rm max}}-\tau \nn\\
		&(\forall n\geq N_2\;s.t.\;\varphi_n\notin C_n), \nn\\
& \forall\signal{h}^*\in \left\{\signal{h}\in\real^{r_n}\mid
\sum_{i=1}^{r_n} h_i f_i^{(n)}\in\Gamma\right\}.
		\end{align}
		Here, $\norm{\signal{E}_n}_2:=\displaystyle\sup_{\sigx\neq \signal{0}}\dfrac{\norm{\signal{E}_n\sigx}_{\Real^{r_{N_0}}}}{\norm{\sigx}_{\Real^{r_{N_0}}}}$.
		Note that the length $r_n$ is fixed for $n\geq N_2(\geq N_1\geq N_0)$.
	\end{enumerate}
\end{assumption}
\begin{remark}[On Assumption \ref{assump1ana}.2]
	Assumption \ref{assump1ana}.2 is reasonable, because it is almost impossible to guarantee convergence in case the dictionary subspace keeps changing indefinitely.  When the input space is compact and the coherence condition is used to construct the dictionary, for instance, the dictionary size remains finite as the time index goes to infinity \cite{richard09}.
\end{remark}
\begin{remark}[On Assumption \ref{assump1ana}.3] 
The assumption requires that there exists a small open ball in $\bigcap_{\substack{n>N_1\\\varphi_n\notin C_n}} C_n$ with respect to $\M_{N_0}$.  In an ideal case where the noise $\nu_n$ is zero and $\psi\in\M_{N_0}$, it is clear that $\psi\in C_n$ for all $n\geq N_0$, because $|\psi(\sigu_n)-d_n|=|\psi(\sigu_n)-\psi(\sigu_n)|=0\leq\rho$ in \refeq{pro2} for any $\rho\geq0$.
Since the evaluation functional over an RKHS is linear, continuous, and hence bounded {\cite[page 9 - Theorem~1]{berlinet2011reproducing}}, 
there exists a constant $M_1>0$ such that $|\hat{f}(\sigu_n)|\leq M_1\norm{\hat{f}}_{\H},\;\forall \hat{f}\in \M_{N_0}$.
Therefore, it follows that $|f(\sigu_n)-\psi(\sigu_n)|\leq M_1\norm{f-\psi}_{\H}$.
Then, for any $\rho>0$, $B_{\epsilon_3}:=\left\{\hat{f}\in\M_{N_0}|\norm{\hat{f}-\psi}_{\H}<\epsilon_3\right\}\subset C_n$ for all $n\geq N_0$, $\epsilon_3:=\frac{\rho}{M_1}$, because  
$\norm{f-\psi}_{\H}<\epsilon_3$ implies $|f(\sigu_n)-\psi(\sigu_n)|\leq M_1\norm{f-\psi}_{\H}<\rho$.  The assumption is thus valid in this case.

In the general case where $\nu_n\neq0$ and/or $\psi\notin\M_{N_0}$, 
it is necessary that $|\nu_n|\leq M_2$, $\left|\psi_{\M_{N_0}^{\perp}}(\sigu_n)\right|\leq M_3$ for some constants $M_2,M_3\in(0,\infty)$, where $\psi_{\M_{N_0}^{\perp}}(\sigu_n):=\psi(\sigu_n)-\psi_{\M_{N_0}}(\sigu_n)$ with $\psi_{\M_{N_0}}:=P_{\M_{N_0}}(\psi)$.
Then, for any $\rho>M_2+M_3$, $B_{\epsilon_4}:=\left\{\hat{f}\in\M_{N_0}|\norm{\hat{f}-\psi_{\M_{N_0}}}_{\H}<\epsilon_4\right\}\subset C_n$ for $\epsilon_4:=\frac{\rho-(M_2+M_3)}{2M_1}$, because $\norm{f-\psi_{\M_{N_0}}}_{\H}<\epsilon_4$ implies 
$|f(\sigu_n)-d_n|\leq|f(\sigu_n)-\psi_{\M_{N_0}}(\sigu_n)|+|\nu_n|+\left|\psi_{\M_{N_0}^{\perp}}(\sigu_n)\right|<  M_1\epsilon_4+M_2+M_3<\rho$.  Therefore, the assumption is still valid in this case.
\end{remark}
\begin{remark}[On Assumption \ref{assump1ana}.5]
For Analytical approach, the metric is fixed (i.e., $\signal{E}_n=\signal{O}$) after the time instant $n=N_0$ 
due to Assumption \ref{assump1ana}.2, and hence the assumption is valid.
	For Finite-sample approach, one can fix the samples to use for taking sample averages
to make $\signal{E}_n=\signal{O}$ after $n=N_2$.
	When Recursive approach is employed,
the approximation becomes tight as $n$ increases under Assumption \ref{assump1ana}.2, and 
thus the assumption is reasonable.
\end{remark}
Now we are ready to prove the following theorem.
\begin{theorem}
	\label{theoremconv} 
	The sequence $(\varphi_n)_{n\in\Natural}$, or $(\sigh_n)_{n\in\Natural}$, generated by Algorithm \ref{modelalgo}
	satisfies the following properties.\\
	(a) Monotone approximation:\\
For any $\signal{h}_n^*\in \{\signal{h}\in\real^{r_n}\mid
\sum_{i=1}^{r_n} h_i f_i^{(n)}\in C_n\}$,
it holds that
	\begin{align}
		\norm{\sigh_n-\sigh^{*}_n}^2_{\sigR_n}-\norm{\sigh_{n+1}-\sigh^{*}_n}^2_{\sigR_n}\geq 0,\;\forall n\in\Natural.
	\end{align}
	(b) Convergence and asymptotic optimality:\\
	The sequence $(\varphi_n)_{n\in\Natural}$ converges to some point $\hat{\varphi}\in\H$, and $\lim_{n\rightarrow\infty}\Theta_n(\varphi_n)=\lim_{n\rightarrow\infty}\Theta_n(\hat{\varphi})=0$,
	if Analytical approach is employed under Assumptions \ref{assump1ana}.1 -- \ref{assump1ana}.3,
	or if Finite-sample/Recursive approach is employed under Assumptions \ref{assump1ana}.1 -- \ref{assump1ana}.5.
\end{theorem}	
\begin{proof}
		\begin{enumerate}[(a)]
		\item The claim is verified by {\cite[Theorem~2(a)]{APSM1}}.  Note that the analysis of APSM is directly applied to the dictionary subspace because Algorithm \ref{modelalgo} updates the current estimate
		within the dictionary subspace at each time instant.
		\item For Analytical approach, the argument in \cite[Theorem~2(a)]{tayu_tsp15} can be applied by considering the $L^2$ space $\H$ instead of an RKHS.
		For Finite-sample/Recursive approach, the same argument of the variable-metric APSM {\cite[Theorem~1(c)]{variable1}} can be applied by considering the fixed dictionary subspace after the dictionary has been well constructed.
		Specifically, {\cite[Assumption~1]{variable1}} is validated by Assumptions \ref{assump1ana}.1, \ref{assump1ana}.2, and \ref{assump1ana}.3, and {\cite[Assumption~2]{variable1}}
		is validated by Assumptions \ref{assump1ana}.4 and \ref{assump1ana}.5 to apply {\cite[Theorem~1(c)]{variable1}}.
	\end{enumerate}
\end{proof}
\begin{remark}[On Theorem \ref{theoremconv}]
	Monotone approximation is significant in the sense that the proposed algorithm can even track the time-varying target function, while convergence is also guaranteed deterministically
		for the time-independent target functions.
	Since the primary focus of the present study is an online learning for possibly time-varying target functions, analyzing the convergence rate is out of the scope.
	The interested readers are referred to the detailed analysis of APSM \cite{APSM1}, which gives the bound of how close the estimate will get to an optimal point at each iteration.
\end{remark}	

\section{Numerical Examples}
\label{sec:numerical}
We first show the decorrelation property of the proposed algorithm.
We then show the efficacy of the proposed algorithm 
in applications to online predictions of two real datasets.
The kernel adaptive filtering toolbox \cite{vanvaerenbergh2013comparative} is used in the experiment.
Throughout the experiments, the set of dictionary data is used
to compute the sample average for Finite-sample approach.

\subsection{Decorrelation Property}
\label{subsec:experimenta}
We compare the eigenvalue spreads of the modified autocorrelation
matrices of the proposed algorithm and the existing multikernel
adaptive filtering algorithms, namely MKNLMS and CHYPASS.
Input vectors are drawn from the i.i.d. uniform distribution within $[-1,1]\subset\Real$, and
Gaussian functions with scale parameters $\s_1:=1.0,~\s_2:=0.5$, and $\s_3:=0.05$ are employed (see Section \ref{subsec:review}).
Dictionary is constructed by the sole use of the coherence condition (i.e, $\epsilon=0$)
with the threshold $\delta=0.8$.
For meaningful comparison,
all the algorithms share the same dictionary which is constructed based on the coherence condition
defined in the Cartesian product of Gaussian RKHSs (see \cite{yukawa_tsp15}).
To avoid numerical errors in computing matrix inverses,
the metric matrix $\sigG_n$ is modified to
$\tilde{\sigG}_n:=\gamma\sigG_n+(1-\gamma)\sigI$,
$\gamma:=0.99$.
The modified autocorrelation matrix $\hat{\sigR}$ is then computed as
$\hat{\sigR}_n:=\tilde{\sigG}_n^{-\frac{1}{2}}\sigR\tilde{\sigG}_n^{-\frac{1}{2}}$, where
$\sigR$ is approximated as $\sigR\approx\frac{1}{N}\sum_{n=1}^{N}\sigf_n\sigf_n^{\T},\;N:=10000$ at every iteration
(see the arguments below Fact \ref{fact:G} in Section \ref{subsec:whyhow}). 
\reffig{eigenevo} plots the evolutions of the eigenvalue spreads of $\hat{\sigR}$ for each algorithm.
One can see that the proposed algorithm attains a smaller
eigenvalue spread of $\hat{\sigR}$, having a better decorrelation property.
For Analytical approach, it works relatively well despite the use of (possibly inappropriate) noninformative distribution for the input vector.
Although Recursive approach shows degradations 
during the initial phase when the dictionary size increases rapidly, 
the eigenvalue spread tends to decrease successfully as the iteration number
increases.
Finite-sample approach shows stable performance
at the expense of high computational complexity of $\mathO(r_n^3)$.
In practice, one may use the selective-update strategy to reduce the complexity (see Section \ref{subsec:selective_update}).
\begin{figure}[t]
	\begin{center}
		\psfrag{Dictionary}{\footnotesize{Finite-sample}}
		\psfrag{MKNLMS}{\footnotesize{MKNLMS}}
		\psfrag{Analytical}{\footnotesize{Analytical}}
		\psfrag{Recursive}{\footnotesize{Recursive}}
		\psfrag{CHYPASS}{\footnotesize{CHYPASS}}
		\includegraphics[clip,width=0.46\textwidth]{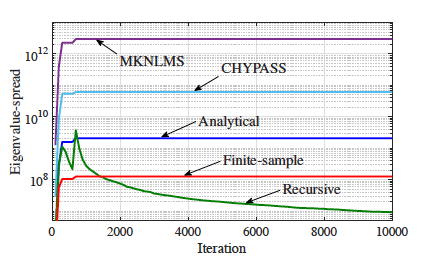}
		\label{fig:eigenevo1}
	\end{center}
	\caption{Evolutions of the eigenvalue spreads of the modified autocorrelation matrices $\hat{\sigR}_n$.  The proposed algorithm shows better decorrelation properties.}
	\label{eigenevo}
\end{figure}
For further clarification, $\hat{\sigR}_n$s for MKNLMS, CHYPASS, and the proposed algorithm (Analytical approach)
are illustrated in \reffig{comparison}.
Here, 
``jet colormap array''
in MATLAB\_R2017b is used for the illustrations.
In particular, we can observe that the off-diagonal elements of $\hat{\sigR}_n$ 
for the proposed algorithm are suppressed better than the other algorithms, as supported quantitatively by
\reffig{eigenevo}.

\begin{figure}[t]
	\begin{minipage}{0.32\hsize}
		\begin{center}
			\includegraphics[clip,width=\textwidth]{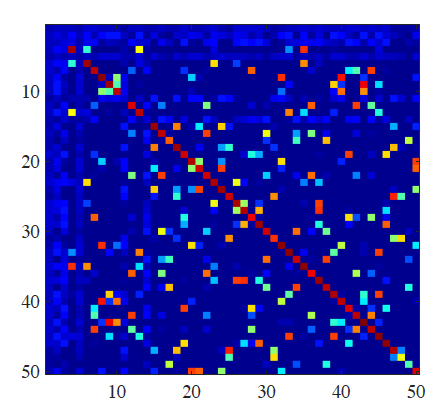}
			\label{fig:Mr}
			\centerline{(a) MKNLMS}
		\end{center}	
	\end{minipage}
	\begin{minipage}{0.32\hsize}
		\begin{center}
			\includegraphics[clip,width=\textwidth]{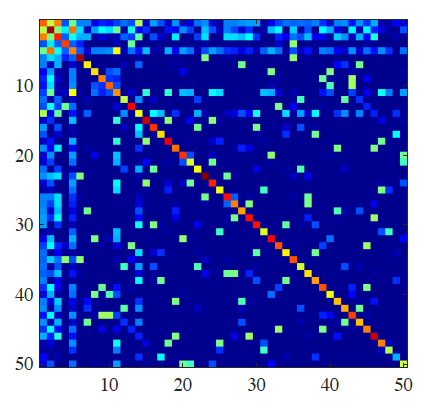}
			\label{fig:Cr}
			\centerline{(b) CHYPASS}
		\end{center}
	\end{minipage}
	\begin{minipage}{0.32\hsize}
		\begin{center}
			\includegraphics[clip,width=\textwidth]{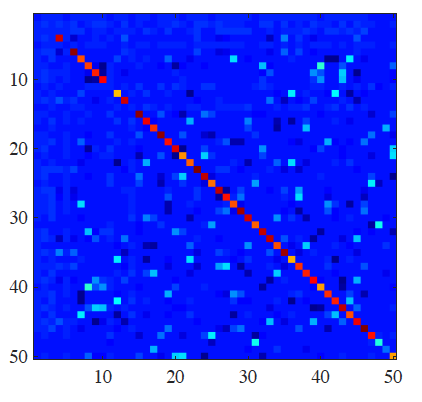}
			\label{fig:Ar}
			\centerline{(c) Proposed}
		\end{center}
	\end{minipage}
\caption{Illustrations of the modified autocorrelation matrices $\hat{\sigR}_n$.
The off-diagonal elements of $\hat{\sigR}_n$ for the proposed algorithm are suppressed
better than the other algorithms.}
	\label{comparison}
\end{figure}
\begin{figure}[t!]
	\begin{center}
		\psfrag{NLMS}{\footnotesize{NLMS}}
		\psfrag{KNLMS}{\footnotesize{KNLMS}}
		\psfrag{HYPASS}{\footnotesize{HYPASS}}
		\psfrag{CHYPASS}{\footnotesize{CHYPASS}}
		\psfrag{Analytical}{\footnotesize{Analytical}}
		\psfrag{Recursive}{\vspace{0.3em}\footnotesize{Recursive}}
		\psfrag{Dictionary}{\footnotesize{Finite-sample}}
		\psfrag{NMSE (dB)}{\footnotesize{NMSE (dB)}}
		\includegraphics[clip,width=0.45\textwidth]{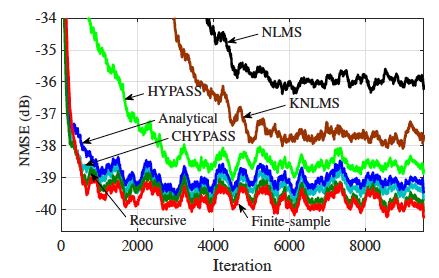}
		\label{fig:mse}
		\centerline{(a) MSE learning curves}
	\end{center}
\vspace{-1em}
	\begin{center}
		\psfrag{KNLMS}{\footnotesize{KNLMS}}
		\psfrag{HYPASS}{\footnotesize{HYPASS}}
		\psfrag{CHYPASS}{\footnotesize{CHYPASS}}
		\psfrag{Analytical}{\footnotesize{Analytical}}
		\psfrag{Recursive}{\footnotesize{Recursive}}
		\psfrag{Dictionary}{\footnotesize{Finite-sample}}
		\includegraphics[clip,width=0.45\textwidth]{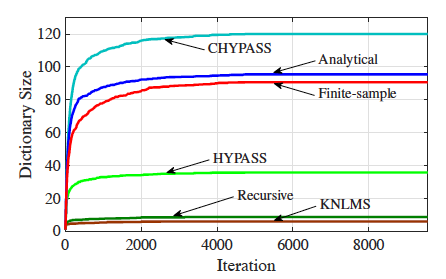}
		\label{fig:dictevo}
		\centerline{(b) Evolutions of the dictionary size}
	\end{center}
\begin{center}
	\psfrag{NLMS}{\footnotesize{NLMS}}
	\psfrag{KNLMS}{\footnotesize{KNLMS}}
	\psfrag{HYPASS}{\footnotesize{HYPASS}}
	\psfrag{CHYPASS}{\footnotesize{CHYPASS}}
	\psfrag{Analytical}{\footnotesize{Analytical}}
	\psfrag{Recursive}{\footnotesize{Recursive}}
	\psfrag{Dictionary}{\footnotesize{Finite-sample}}
	\includegraphics[clip,width=0.45\textwidth]{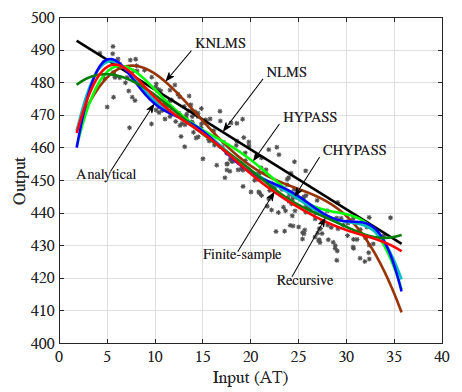}
	\label{fig:elecsys}
	\centerline{(c) Estimate of each algorithm and the target values}
\end{center}
	\caption{Results of the regression task on the online prediction of the electrical power output.  The estimators are trained until $4783$ iterations in an online fashion
	and the trained estimator is used for the subsequent iterations.}
	\label{real1result}
\end{figure}
\subsection{Online prediction of electrical power output}
\label{subsec:experimentb}
We consider the online prediction of electrical power output analyzed in \cite{elecpower1,elecpower2}.
The target variable, namely the full load electrical power output, depends highly on ambient temperature (AT).
Because AT is strongly correlated with the target variable and can individually predict the target variable \cite{elecpower1,elecpower2},
AT is employed as a sole input variable in the present experiment for the comparison purpose.
In \cite{elecpower1,elecpower2}, different machine learning regression methods are compared to each other
in terms of the root mean squared error (RMSE).
The same dataset and problem settings are used to compare the RMSE performance of the proposed algorithm with
linear NLMS, KNLMS, HYPASS, CHYPASS, and the machine learning regression methods analyzed in \cite{elecpower1,elecpower2}.
Note here that the proposed algorithm is designed for online learning, while those analyzed in \cite{elecpower1,elecpower2} are batch methods.
It is observed in \cite{elecpower1,elecpower2}, that AT affects the target variable more than the other variables, and that
the model trees rules (M5R) achieves the lowest RMSE $5.085$ among $15$ machine learning regression methods.

Following \cite{elecpower1,elecpower2}, 5 $\x$ 2 cross-validation is employed, i.e., datasets are equally partitioned into two sets with the same size and each of the sets is trained to validate the other (2-fold cross-validation), and it is repeated five times by shuffling the datasets.  For the comparison purposes, the same five shuffled data as those in \cite{elecpower1,elecpower2} are used.  The RMSEs over the test set for $5\x2=10$ runs are then averaged to obtain the final results.
Note that the estimator is trained in an online fashion with the first half of the dataset, and the trained estimator is applied to the other half.
To choose the best parameters for each algorithm, we first use a coarse search to find rough regions
of {\it good} parameters, and then exploit a fine random search \cite{randomsearch} of $100$ combinations to find the best parameters achieving the best RMSE averaged over the $10$ runs.
For the nonlinear estimators, Gaussian functions are employed with
 fixed scale parameters because an advantage of using multiple Gaussian functions
	is that no elaborative parameter tuning is needed.  For the monokernel methods, the best scale parameter $\sigma$ is selected.

\reftab{sumparam} summarizes the parameter settings and the means and standard deviations of RMSEs over the $10$ runs including those of the batch methods.
It is observed that Finite-sample approach achieves lower RMSE than the batch methods excluding the top-two methods
(M5R and the model trees regression).
Moreover, it can be observed that the use of multiple Gaussian functions leads to significantly better results than their monokernel counterparts.
The normalized MSE (NMSE) learning curves averaged over the $5\x2=10$ runs are smoothed and plotted in \reffig{real1result}(a), and the evolutions of dictionary size are plotted in \reffig{real1result}(b).
\reffig{real1result}(c) shows an instance of the estimate of each algorithm over the test set of the final run and the target values for $235$ input data (AT) selected from the test set of the final run.
\begin{table}[t]
	\begin{center}
		\caption{Summary of the best parameter settings and the RMSE performances
for electrical power output data}
		\label{sumparam}
		\begin{tabular}{|c|c|c|}\hline
			Algorithm (parameters) & Type& RMSE \\\hline
			model trees rules &Batch& $5.085$ \\\hline
			model trees regression &Batch& $5.086$ \\\hline
			\textbf{Finite-sample} &&\\
			($\lambda_n=0.0334, \delta=0.9988, \epsilon=0, \rho=0$ &\textbf{Online}& \boldmath$5.1436$ \\
			$\s_1:=40,\s_2:=25,\s_3:=15,\s_4:=5$&&\boldmath$\pm 0.0230$\\
			$\gamma=0.999,s_n=7,\gamma_{\rm update}=1.00\x10^{-8}$)&&\\\hline
			bagging REP tree &Batch& $5.208$ \\\hline
			reduced error pruning &Batch& $5.229$ \\\hline
			\textbf{Recursive} &&\\
			($\lambda_n=0.0804, \delta=0.9313, \epsilon=0, \rho=0$ &\textbf{Online}& \boldmath$5.2327$ \\
			$\s_1:=40,\s_2:=25,\s_3:=15,\s_4:=5$&&\boldmath$\pm 0.0641$\\
			$\gamma=0.999,s_n=7,\gamma_{\rm update}=1.00\x10^{-8}$)&&\\\hline
			KStar&Batch & $5.381$ \\\hline
			pace regression&Batch & $5.426$ \\\hline
			linear regression&Batch & $5.426$ \\\hline
			simple linear regression&Batch & $5.426$ \\\hline
			CHYPASS &&\\
			($\lambda_n=0.1749, \delta=0.9976, \epsilon=0, \rho=0$&Online & $5.4262$ \\
			$\s_1:=40,\s_2:=25,\s_3:=15,\s_4:=5$&&$\pm 0.1842$\\
			$\gamma=0.999,s_n=7,\gamma_{\rm update}=1.00\x10^{-8}$)&&\\\hline
			support vector poly kernel regression&Batch & $5.433$ \\\hline
			least median square&Batch & $5.433$ \\\hline
			\textbf{Analytical} &&\\
			($\lambda_n=0.2183, \delta=0.9981, \epsilon=0, \rho=0$&\textbf{Online} & \boldmath$5.5680$ \\
			$\s_1:=40,\s_2:=25,\s_3:=15,\s_4:=5$&&\boldmath$\pm 0.2496$\\
			$\gamma=0.999,s_n=7,\gamma_{\rm update}=1.00\x10^{-8}$)&&\\\hline
			HYPASS &&\\
			($\lambda_n=0.3784, \delta=0.9971, \epsilon=0$& Online& $5.8733$ \\
			$\rho=0, \s=9.4857, \gamma=0.999$)&&$\pm 0.3623$\\
			$s_n=7,\gamma_{\rm update}=1.00\x10^{-8}$)&&\\\hline
			additive regression&Batch & $5.933$ \\\hline
			IBk linear NN search&Batch & $6.377$ \\\hline
			multi layer perceptron&Batch & $6.483$ \\\hline
			KNLMS &&\\
			($\lambda_n=0.5720, \delta=0.9481, \epsilon=0$&Online & $6.5152$ \\
			$\rho=0, \s=15.0643, \gamma=0.999$)&&$\pm 0.5155$\\
			$s_n=7,\gamma_{\rm update}=1.00\x10^{-8}$)&&\\\hline
			NLMS & &  \\
			($\lambda_n=0.8435$&Online & $7.4680$ \\
			$\gamma_{\rm update}=1.00\x10^{-8}, \rho=0$)& &$\pm 2.1496$  \\\hline
			radial basis function neural network&Batch & $7.501$\\\hline
			locally weighted learning &Batch& $8.005$ \\\hline
		\end{tabular}
	\end{center}
\end{table}
\begin{figure}[t!]
	\begin{center}
		\psfrag{NLMS}{\footnotesize{NLMS}}
		\psfrag{KNLMS}{\footnotesize{KNLMS}}
		\psfrag{KRLS-T}{\footnotesize{KRLS-T}}
		\psfrag{EKF}{\footnotesize{EKF}}
		\psfrag{HYPASS}{\footnotesize{HYPASS}}
		\psfrag{CHYPASS}{\footnotesize{CHYPASS}}
		\psfrag{Analytical}{\hspace{-0.2em}\footnotesize{Analytical}}
		\psfrag{Recursive}{\footnotesize{Recursive}}
		\psfrag{Dictonary}{\footnotesize{Finite-sample}}
		\psfrag{NMSE(dB)}{\footnotesize{NMSE (dB)}}
		\includegraphics[clip,width=0.45\textwidth]{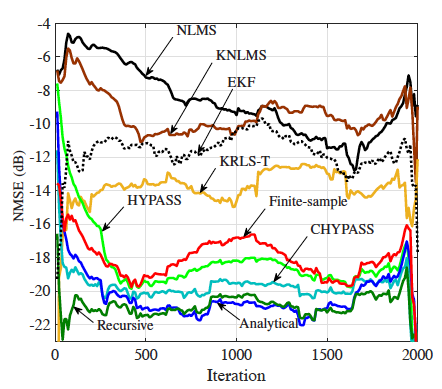}
		\label{fig:mse21}
		\centerline{(a) NMSE learning curves}
	\end{center}
\vspace{-1em}
	\begin{center}
		\psfrag{KNLMS}{\footnotesize{KNLMS}}
		\psfrag{KRLS-T}{\footnotesize{KRLS-T}}
		\psfrag{HYPASS}{\footnotesize{HYPASS}}
		\psfrag{CHYPASS}{\footnotesize{CHYPASS}}
		\psfrag{Analytical}{\footnotesize{Analytical}}
		\psfrag{Recursive}{\footnotesize{Recursive}}
		\psfrag{Dictionary}{\footnotesize{Finite-sample}}
		\includegraphics[clip,width=0.45\textwidth]{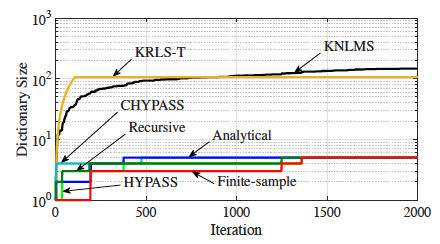}
		\label{fig:dictevo2}
		\centerline{(b) Evolutions of dictionary size}
	\end{center}
	\caption{Results of time-series data prediction of GPS measurement.  The proposed algorithm performs better than the extended Kalman Filter.}
	\label{real2result}
\end{figure}
\subsection{Online Prediction of GPS Measurements}
\label{subsec:experimentc}
We use the real trajectory data of GPS positions, the dynamics of a true vehicle,
and the simulated pseudo range measurements \cite{skog2011time} given by
\begin{align}
d^{(i)}_n=\norm{\sigp^{(i)}_n-\sigp^{({\rm rec})}_n}_{\Real^3}+c\Delta t_n+\nu_n,\;i\in\overline{1,7},
\end{align}
where $\sigp^{(i)}_n\in\Real^3$ is the position of the $i$-th GPS, $\sigp^{({\rm rec})}_n\in\Real^3$ is the position of the vehicle at time $n$, 
and $c$ is the speed of light, $\Delta t_n$ is the clock offset, and $\nu_n$ is the zero-mean Gaussian noise with variance $4.0$.
Given the available measurements of the vehicle with nonmaneuvering motion,
the task is to predict the next pseudo range measurement of the first GPS.
In this experiment, we compare the NMSE performance of the proposed algorithm with
NLMS, KNLMS, KRLS-T, HYPASS, CHYPASS, and EKF \cite{EKF1}.
For EKF, seven available GPS measurements are used to estimate the vehicle position and the next GPS measurements based on the nonmaneuvering motion model presented in \cite{GPSsurvey}.
Noise models used in EKF are tuned by using the $3\s$ confidence interval.
The other algorithms exploit less information than EKF and use only the measurement of the first GPS.  The next measurement $d^{(1)}_{n+1}$
is estimated with $\sigu_n:=[d^{(1)}_{n},d^{(1)}_{n-1},d^{(1)}_{n-2}]^{\T}$.

To choose the best parameters for each algorithm, we again use the coarse-fine random search of $100$ combinations described in Section \ref{subsec:experimentb}.
For the nonlinear estimators, Gaussian functions are employed with fixed
scale parameters.  For the monokernel methods, the best scale parameter $\sigma$ is selected.
The coherence threshold $\delta$ is
tuned so that the final dictionary sizes become the same among HYPASS, CHYPASS and the proposed algorithm.
The regularization parameter $\gamma_{\rm update}$ is tuned carefully only for KRLS-T because of sensitivity.

\reftab{sumparam2} summarizes the parameter settings.  Here, $M,~\xi$, $\gamma_{\rm update}$ are the budget, the forgetting factor,
and the regularization parameter  for KRLS-T, respectively.
The MSE learning curves are plotted in \reffig{real2result}(a), and the evolutions of dictionary size are plotted in \reffig{real2result}(b).
It can be observed that HYPASS, CHYPASS, and the proposed algorithm outperform EKF despite the use of less information.
Finite-sample approach performs worse than Analytical/Recursive approach because of the small dictionary size.
\begin{table}[t]
	\begin{center}
		\caption{Summary of the best parameter settings for GPS data}
		\label{sumparam2}
		\begin{tabular}{|c|c|}\hline
			Algorithm & parameters  \\\hline
			NLMS & $\lambda_n=0.0017,\gamma_{\rm update}=1.00\x10^{-8},\rho=0$  \\\hline
			&$\lambda_n=0.0055, \s=8.9761,\delta=0.8,\epsilon=0.01,\rho=0$ \\
			KNLMS & $s_n=7, \gamma=0.999,\gamma_{\rm update}=1.00\x10^{-8}$ \\\hline
			& $M=106, \s=23.2953$  \\
			KRLS-T & $\xi=0.8373, \gamma_{\rm update}=5.5599\x10^{-6}$ \\\hline
			&$\lambda_n=0.4235, \s=43.8043,\delta=0.8,\epsilon=0.01,\rho=0$\\
			HYPASS & $s_n=7, \gamma=0.999,\gamma_{\rm update}=1.00\x10^{-8}$ \\\hline
			&$\lambda_n=0.4293,\delta=0.76,\epsilon=0.01,\rho=0$\\
			CHYPASS & $s_n=7, \gamma=0.999,\gamma_{\rm update}=1.00\x10^{-8}$  \\
			&$\s_1:=55,\s_2:=50,\s_3:=45,\s_4:=40$\\\hline
			&$\lambda_n=0.6399,\delta=0.8,\epsilon=0.01,\rho=0$ \\
			Analytical & $s_n=7, \gamma=0.999,\gamma_{\rm update}=1.00\x10^{-8}$ \\
			&$\s_1:=55,\s_2:=50,\s_3:=45,\s_4:=40$\\\hline
			&$\lambda_n=0.6320,\delta=0.9880,\epsilon=0.01,\rho=0$\\
			Finite-sample & $s_n=7, \gamma=0.999,\gamma_{\rm update}=1.00\x10^{-8}$\\
			&$\s_1:=55,\s_2:=50,\s_3:=45,\s_4:=40$\\\hline
			&$\lambda_n=1.9842,\delta=0.9888,\epsilon=0.01,\rho=0$\\
			Recursive & $s_n=7, \gamma=0.999,\gamma_{\rm update}=1.00\x10^{-8}$ \\
			&$\s_1:=55,\s_2:=50,\s_3:=45,\s_4:=40$\\\hline
		\end{tabular}
	\end{center}
\end{table}

\section{Conclusion}
\label{sec:conclusion}

The online learning paradigm presented in this paper is a significant extension of 
the conventional kernel adaptive filtering framework from RKHS
to the space $L^2(\Real^L,\di\mu)$ which has no reproducing kernel
and which induces the best geometry
in the sense of decorrelation.
The proposed algorithm was built upon the fact that
the reproducing kernel of the dictionary subspace can be obtained
in terms of the Gram matrix.
Three approaches to computing the Gram matrix were presented.
A remarkable difference from kernel adaptive filtering is that
the whole space $L^2(\Real^L,\di\mu)$ has no reproducing kernel.
In $L^2(\Real^L,\di\mu)$, 
the MMSE estimator gives the best approximation
of the target nonlinear function in the dictionary subspace
in contrast to the case of kernel adaptive filtering.
Also, the ALD condition in $L^2(\Real^L,\di\mu)$
ensures a lower bound of the amount of the MMSE reduction
due to the newly entering atom.
The selective-update strategy was presented to reduce the computational complexity.
The analysis was presented to show the monotone approximation, 
asymptotic optimality, and convergence of the proposed algorithm for the full-updating case.
The numerical examples demonstrated
the efficacy of the proposed algorithm using the selective-update strategy for two real datasets, showing
its superior performance to the extended Kalman filter and
comparable performance with 
the best batch machine-learning method that was tested.
We finally emphasize that the proposed paradigm can be extended
straightforwardly to any other functional spaces
as long as the Gram matrix can be computed efficiently.


\begin{appendices}
	\numberwithin{equation}{section}
\renewcommand{\theequation}{\thesection.\arabic{equation}}
\section{Proof of Proposition \ref{prp3}}
\label{appprp3}
	Let $P_{\M_n}(\psi):=\sum_{i=1}^{r_n}h_i f_i^{(n)},\; h_i\in\Real$, then
	the coefficient vector $\sigh:=\left[h_1,h_2,\cdots,h_{r_n}\right]^{\T}\in\Real^{r_n}$ 
is characterized by the following normal equation \cite{luenberger}:
	\begin{equation}
		\sigR\sigh=\signal{b}, \label{best}
	\end{equation}
	where $\sigR$ is the Gram matrix of the dictionary (see Fact \ref{fact:G}),
 and $\signal{b}:=\left[\inpro{f_1,\psi}_{\H},\inpro{f_2,\psi}_{\H},\cdots,\inpro{f_{r_n},\psi}_{\H}\right]^{\T}
\in\Real^{r_n}$.
Here, it holds that $\signal{b}=\signal{p}(:=E\left[\signal{f}_n(\sigu_n)d_n\right])$ because
	\begin{align}
		&E\left[f_i(\sigu_n)d_n\right]=E\left[f_i(\sigu_n)(\psi(\sigu_n)+\nu_n)\right]\nonumber\\
		&=E\left[f_i(\sigu_n)\psi(\sigu_n)\right]+E\left[f_i(\sigu_n)\nu_n\right]\nonumber\\
		&=E\left[f_i(\sigu_n)\psi(\sigu_n)\right]+0\nonumber\\
		&=\int_{\Real^L}f_i(\sigu)\psi(\sigu)p(\sigu)d\sigu=\inpro{f_i,\psi}_{\H}.
	\end{align}
Hence, \refeq{best} is equivalent to 
		$\sigR\sigh=\signal{p}$,
which is nothing but the Wiener-Hopf equation derived directly from \refeq{cost1} to obtain the MMSE estimator.

\section{Proof of Proposition \ref{prp1}}
\label{appprp1}
	It is clear that $\ka{\cdot,\sigu}\in{\rm span}\mathcal{D}$ 
for any $\sigu\in\Real^L$.
	By definition of $\inpro{\cdot,\cdot}_{\H}$, it can be readily verified that
	\begin{align}
		&\inpro{\ka{\cdot,\sigu},\ka{\cdot,\sigv}}_{\H}\nonumber\\
		&=\int_{\Real^L}\underbrace{\signal{f}^{\T}(\sigu)\sigG^{-1}\signal{f}(\sigw)}_{\ka{\sigw,\sigu}}
		\underbrace{\signal{f}^{\T}(\sigw)\sigG^{-1}\signal{f}(\sigv)}_{\ka{\sigw,\sigv}}\di\mu(\sigw)\nonumber\\
		&=\signal{f}^{\T}(\sigu)\sigG^{-1}\underbrace{\int_{\Real^L}\signal{f}(\sigw)\signal{f}^{\T}(\sigw)\di\mu(\sigw)}_{\sigG}
		\sigG^{-1}\signal{f}(\sigv)\nonumber\\
		&=\ka{\sigu,\sigv}.
	\end{align}
	For any $\sigu\in\Real^L$ and $\phi:=\sum_{i=1}^r\alpha_if_i,\alpha_i\in\Real$,
	the reproducing property holds:
	\begin{align}
		&\inpro{\phi,\ka{\cdot,\sigu}}_{\H}=\int_{\Real^L}\phi(\sigw)\signal{f}^{\T}(\sigw)\sigG^{-1}\signal{f}(\sigu)\di\mu(\sigw) \nonumber\\
		&=\sum_{i=1}^r\alpha_i\underbrace{\int_{\Real^L}f_i(\sigw)\signal{f}^{\T}(\sigw)\di\mu(\sigw)}_{\sigg_i^{\T}:=[g_{i,1},g_{i,2},\cdots,g_{i,r}]}\sigG^{-1}\signal{f}(\sigu) \nonumber\\
		&
		=\sum_{i=1}^r\alpha_i\underbrace{\signal{e}_i^{\T}\signal{f}(\sigu)}_{f_i(\sigu)}=\phi(\sigu),
	\end{align}	
	where $\left\{\signal{e}_i\right\}_{i=1}^r$ is the standard basis of $\Real^r$.
\section{Proof of Proposition \ref{prp2}}
\label{appprp2}
\begin{enumerate}[(a)]
	\item
The inner product can be computed as follows:
\begin{align}
	&\hspace{-3em}\inpro{\kappa_p(\cdot,\sigu),\kappa_q(\cdot,\sigv)}_{\H}=
	\frac{1}{(2\pi\s_p^2)^{L/2}}\frac{1}{(2\pi\s_q^2)^{L/2}}\frac{1}{(2\pi\s^2)^{L/2}} \nonumber \\
	&\hspace{-3em}\int\exp\left[-\left(\underbrace{\frac{\norm{\sigu-\sigw}_{\Real^L}^2}{2\s_p^2}+\frac{\norm{\sigv-\sigw}_{\Real^L}^2}{2\s_q^2}+\frac{\norm{\sigw}_{\Real^L}^2}{2\s^2}}_{A(\sigw)}
	\right)\right] \di\sigw. \label{C1}
\end{align}
Here, $A(\sigw)=\frac{\norm{\sigu}_{\Real^L}^2}{2\s_p^2}+\frac{\norm{\sigw}_{\Real^L}^2}{2\s_p^2}-\frac{2\inpro{\sigu,\sigw}_{\Real^L}}{2\s_p^2}
+\frac{\norm{\sigv}_{\Real^L}^2}{2\s_q^2}+\frac{\norm{\sigw}_{\Real^L}^2}{2\s_q^2}
-\frac{2\inpro{\sigv,\sigw}_{\Real^L}}{2\s_q^2}
+\frac{\norm{\sigw}_{\Real^L}^2}{2\s^2}
=\frac{\norm{\sigu}_{\Real^L}^2}{2\s_p^2}+\frac{\norm{\sigv}_{\Real^L}^2}{2\s_q^2}+\frac{\alpha}{2}\left(\norm{\sigw}_{\Real^L}^2-2\inpro{\frac{\sigz}{\alpha},\sigw}_{\Real^L}\right), 
$
where $\alpha:=\frac{1}{\s_p^2}+\frac{1}{\s_q^2}+\frac{1}{\s^2}>0$, and
$\sigz:=\frac{\sigu}{\s_p^2}+\frac{\sigv}{\s_q^2}$, from which it follows that
\begin{align}
	&\int\exp(-A(\sigw))\di \sigw\nn\\
	&=\exp\left\{-\left(\underbrace{-\frac{\norm{\sigz}_{\Real^L}^2}{2\alpha}+\frac{\norm{\sigu}_{\Real^L}^2}{2\s_p^2}+\frac{\norm{\sigv}_{\Real^L}^2}{2\s_q^2}}_{B}\right)\right\} \nonumber\\
	&\underbrace{\int\exp\left(-\frac{\norm{\sigw-\frac{\sigz}{\alpha}}_{\Real^L}^2}{2\frac{1}{\alpha}}\right) \di\sigw}_{=(2\pi\frac{1}{\alpha})^{L/2}}. \label{C2}
\end{align}
Here,
$B=-\frac{1}{2\alpha}
	\left(\frac{\norm{\sigu}_{\Real^L}^2}{\s_p^4}+\frac{\norm{\sigv}_{\Real^L}^2}{\s_q^4}+\frac{2\inpro{\sigu,\sigv}_{\Real^L}}{\s_p^2\s_q^2}\right)
	+\frac{\norm{\sigu}_{\Real^L}^2}{2\s_p^2}+\frac{\norm{\sigv}_{\Real^L}^2}{2\s_q^2}
	=\frac{\s^2\norm{\sigu-\sigv}_{\Real^L}^2-\frac{\s^2\s_p^2+\s^2\s_q^2}{\s_p^2}\norm{\sigu}_{\Real^L}^2-\frac{\s^2\s_q^2+\s^2\s_p^2}{\s_q^2}\norm{\sigv}_{\Real^L}^2}{2(\s^2\s_p^2+\s^2\s_q^2+\s_p^2\s_q^2)}
	+\frac{\norm{\sigu}_{\Real^L}^2}{2\s_p^2}
	+\frac{\norm{\sigv}_{\Real^L}^2}{2\s_q^2}
	=\frac{\s^2\norm{\sigu-\sigv}_{\Real^L}^2+\s_q^2\norm{\sigu}_{\Real^L}^2+\s_p^2\norm{\sigv}_{\Real^L}^2}{2(\s^2\s_p^2+\s^2\s_q^2+\s_p^2\s_q^2)}
$.
It thus follows by \refeq{C1} and \refeq{C2} that
\begin{align}
 &\hspace{-2em}\inpro{\kappa_p(\cdot,\sigu),\kappa_q(\cdot,\sigv)}_{\H}=\frac{1}{(2\pi)^{L}}\frac{1}{\left(\s^2\s_p^2+\s^2\s_q^2+\s_p^2\s_q^2\right)^{L/2}}  \nonumber \\
  &\hspace{-2em}\exp\left(-\frac{\s^2\norm{\sigu-\sigv}_{\Real^L}^2+\s_q^2\norm{\sigu}_{\Real^L}^2+\s_p^2\norm{\sigv}_{\Real^L}^2}{2(\s^2\s_p^2+\s^2\s_q^2+\s_p^2\s_q^2)}\right). \label{case11}
\end{align}
\item
Since $\di\mu(\sigu)=\di\sigu$, it follows that
\begin{align}
	&\inpro{\kappa_p(\cdot,\sigu),\kappa_q(\cdot,\sigv)}_{\H}=\frac{1}{(2\pi\s_p^2)^{L/2}}\frac{1}{(2\pi\s_q^2)^{L/2}}\nn\\
	&\int\exp\left[-\left(\underbrace{\frac{\norm{\sigu-\sigw}_{\Real^L}^2}{2\s_p^2}+\frac{\norm{\sigv-\sigw}_{\Real^L}^2}{2\s_q^2}}_{C(\sigw)}\right)\right]\di \sigw. \label{C3}
\end{align}
Here, $C(\sigw)=\frac{\norm{\sigu}_{\Real^L}^2}{2\s_p^2}+\frac{\norm{\sigw}_{\Real^L}^2}{2\s_p^2}-\frac{2\inpro{\sigu,\sigw}_{\Real^L}}{2\s_p^2}
+\frac{\norm{\sigv}_{\Real^L}^2}{2\s_q^2}+\frac{\norm{\sigw}_{\Real^L}^2}{2\s_q^2}
-\frac{2\inpro{\sigv,\sigw}_{\Real^L}}{2\s_q^2}
=\frac{\norm{\sigu}_{\Real^L}^2}{2\s_p^2}+\frac{\norm{\sigv}_{\Real^L}^2}{2\s_q^2}+\frac{\beta}{2}\left(\norm{\sigw}_{\Real^L}^2-2\inpro{\frac{\sigz}{\beta},\sigw}_{\Real^L}\right),
$
where $\beta:=\frac{1}{\s_p^2}+\frac{1}{\s_q^2}>0$, and $\sigz:=\frac{\sigu}{\s_p^2}+\frac{\sigv}{\s_q^2}$.
Therefore, \refeq{C3} becomes
\begin{align}
&\hspace{-1.3em}\inpro{\kappa_p(\cdot,\sigu),\kappa_q(\cdot,\sigv)}_{\H}\nn\\
&\hspace{-1.3em}=\frac{1}{(2\pi\s_p^2)^{L/2}}\frac{1}{(2\pi\s_q^2)^{L/2}} \int\exp\left(-\frac{\norm{\sigw-\frac{\sigz}{\beta}}_{\Real^L}^2}{2\frac{1}{\beta}}\right) \di\sigw \nonumber \\
&\exp\left\{-\left(-\frac{\norm{\sigz}_{\Real^L}^2}{2\beta}+\frac{\norm{\sigu}_{\Real^L}^2}{2\s_p^2}+\frac{\norm{\sigv}_{\Real^L}^2}{2\s_q^2}\right)\right\} \nonumber\\
&\hspace{-1.3em}=\frac{1}{(2\pi(\s_p^2+\s_q^2))^{L/2}}
\exp\left(-\frac{\norm{\sigu-\sigv}_{\Real^L}^2}{2(\s_p^2+\s_q^2)}\right). \label{C4}
\end{align}
\end{enumerate}
We mention that the result in \refeq{inprocompute} is also obtained in the Gaussian RKHS by taking the limit of its scale parameter towards zero in
{\cite[Equation~(25)]{tanaka2011}} because Gaussian RKHSs have a nested structure \cite{vert2006consistency,steinwart2006explicit}.
\section{Proof of Proposition \ref{prp4}}
\label{appprp4}
	By the independence assumptions and the definition of $\norm{\cdot}_{\H}$, we have
	\begin{align}
	\hspace*{-2em}	\Ex{d_n-\psi^{*}_{\M_n}(\sigu_n)}^2
&=\Ex{\psi(\sigu_n)-\psi^{*}_{\M_n}(\sigu_n)}^2+E(\nu_n^2)\nonumber\\
\hspace*{-2em}		&=\norm{\psi-\psi^{*}_{\M_n}}^2_{\H}+E(\nu_n^2)\\
\hspace*{-2em}		\Ex{d_n-\psi^{*}_{\M_{n+1}}(\sigu_n)}^2&=\norm{\psi-\psi^{*}_{\M_{n+1}}}^2_{\H}+E(\nu_n^2).
	\end{align}
By Pythagorean theorem and the assumed ALD condition, it follows that
	\begin{align}
\hspace*{-2em}		&\Delta{\rm MMSE} 
		=\norm{\psi-\psi^{*}_{\M_n}}^2_{\H}-\norm{\psi-\psi^{*}_{\M_{n+1}}}^2_{\H}\nonumber\\
\hspace*{-2em}		&=\norm{\psi^{*}_{\M_{n+1}}-\psi^{*}_{\M_n}}^2_{\H} =
\norm{\psi^{*}_{\M_{n+1}}-P_{\M_n}P_{\M_{n+1}}(\psi)}^2_{\H}\nn\\
\hspace*{-2em}		&=\norm{\psi^{*}_{\M_{n+1}}-P_{\M_n}(\psi^{*}_{\M_{n+1}})}^2_{\H} \nn\\
\hspace*{-2em}		&=\norm{h^*_{\signal{u}_n}f_{\signal{u}_n}-P_{\M_n}(h^*_{\signal{u}_n}f_{\signal{u}_n})}_{\H}^2\nn\\
\hspace*{-1.5em}		&=(h^*_{\signal{u}_n})^2\norm{f_{\signal{u}_n}-P_{\M_n}(f_{\signal{u}_n})}^2_{\H}\geq (h^*_{\signal{u}_n})^2\norm{f_{\signal{u}_n}}_{\H}^2\eta.
		\label{msecont}
	\end{align}
\section{RLS as iterative variable-metric projection method}
\label{subsec:recursiveapp}
We first write down a variant of RLS for reference:
\begin{equation}
	\sigx_{n+1}=\sigx_n+\lambda_n\frac{d_n-\sigu_n^{\T}\sigx_n}{\sigu_n^{\T}\sigR_n^{-1}\sigu_n}\sigR_n^{-1}\sigu_n, \label{RLS}
\end{equation}
where $\sigx_n\in\Real^L$ is the coefficient vector, $\sigR_n=\sigR_{n-1}+\sigu_n\sigu_n^{\T}$ and $\lambda_n=\frac{\sigu_n^{\T}\sigR_n^{-1}\sigu_n}{\sigu_n^{\T}\sigR_n^{-1}\sigu_n+1}$.
Although RLS in \refeq{RLS} iteratively minimizes
\begin{equation}
J(\sigh)=\sum_{i=1}^n\left(d_i-\sigu_i^{\T}\sigh\right)^2, \label{rlscost}
\end{equation}
it can also be viewed as a variable-metric projection with the time-varying step size $\lambda_n$ 
under the framework of \cite{variable1} as pointed out in \cite{yukawalecture}.
\end{appendices}

%
%
%

%
 \begin{biography}[{\includegraphics[width=1in,height=1.25in,clip,keepaspectratio]{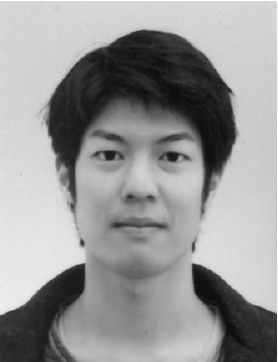}}]{Motoya Ohnishi}
 (S'15) received the B.S. degree in Electronics and Electrical Engineering
 from Keio University, Tokyo, Japan, in 2016.
 He is currently
 working toward the M.S.
 degrees both in Electronics and Electrical Engineering from Keio University, Tokyo, Japan, and Electrical Engineering from
 KTH Royal Institute of Technology, Stockholm, Sweden.
 He was a research assistant at the Department of Automatic Control at KTH Royal Institute of Technology, and was a visiting researcher at GRITSlab
 at Georgia Institute of Technology, Atlanta, USA, in 2017, and is currently a research assistant at RIKEN AIP center, Tokyo, Japan.
 His research interests include mathematical signal processing,
 machine learning, and robotics.
 \end{biography}

 \begin{biography}[{\includegraphics[width=1in,height=1.25in,clip,keepaspectratio]{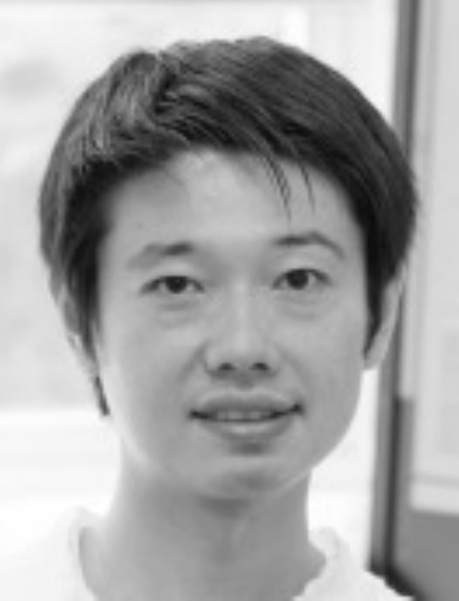}}]{Masahiro Yukawa}
 	(S'05--M'06) received the B.E., M.E., and
 	Ph.D. degrees from Tokyo Institute of Technology
 	in 2002, 2004, and 2006, respectively. He studied
 	as Visiting/Guest Researcher with the University
 	of York, U.K., for half a year, and with the
 	Technical University of Munich, Germany, for four
 	months. He worked with RIKEN, Japan, as Special
 	Postdoctoral Researcher for three years, and with
 	Niigata University, Japan, as Associate Professor
 	for another three years. in 2016, he studied with Machine Learning Group of the Technical University of Berlin as Visiting Professor.
 	He is currently an Associate Professor with the Department of Electronics and
 	Electrical Engineering, Keio University, Japan. He has been Associate
 	Editor for the IEEE TRANSACTIONS ON SIGNAL PROCESSING (since
 	2015), Multidimensional Systems and Signal Processing (2012–2016), and
 	the IEICE Transactions on Fundamentals of Electronics, Communications and
 	Computer Sciences (2009–2013). His research interests include mathematical
 	adaptive signal processing, convex/sparse optimization, and machine learning.
 	Dr. Yukawa was a recipient of the Research Fellowship of the Japan Society
 	for the Promotion of Science (JSPS) from April 2005 to March 2007.He
 	received the Excellent Paper Award and the Young Researcher Award from
 	the IEICE in 2006 and in 2010, respectively, the Yasujiro Niwa Outstanding
 	Paper Award in 2007, the Ericsson Young Scientist Award in 2009, the
 	TELECOM System Technology Award in 2014, the Young Scientists' Prize,
 	the Commendation for Science and Technology by the Minister of Education,
 	Culture, Sports, Science and Technology in 2014, the KDDI Foundation
 	Research Award in 2015, and the FFIT Academic Award in 2016. He is
 	a member of the IEICE.
 \end{biography}
\end{document}